%% file: main_v7.tex
\documentclass[aps,prx,reprint,superscriptaddress,longbibliography,floatfix]{revtex4-2}
\pdfoutput=1
\usepackage[T1]{fontenc}
\usepackage[utf8]{inputenc}
\usepackage{amsmath,amssymb,mathtools,booktabs,array,microtype,graphicx}
\usepackage[colorlinks=true,hyperindex,breaklinks,
            linkcolor=blue,urlcolor=blue,citecolor=blue]{hyperref}
\usepackage{xurl}
\usepackage{multirow}
\usepackage{tikz}
\usetikzlibrary{arrows.meta}
\hypersetup{pdftitle={Designing Group-Valued Codes with Full Regular Low-Weight Bases},pdfauthor={Jong Yeon Lee}}
\newtheorem{theorem}{Theorem}
\newtheorem{proposition}[theorem]{Proposition}

\newtheorem{definition}[theorem]{Definition}
\newenvironment{proof}{\par\noindent\textit{Proof.}\ }{\hfill$\square$\par}
\newcommand{\F}{\mathbb F_2}
\newcommand{\R}{\mathcal R}
\newcommand{\LX}{\mathcal L_X}
\newcommand{\LZ}{\mathcal L_Z}
\DeclareMathOperator{\im}{im}
\DeclareMathOperator{\rank}{rank}
\DeclareMathOperator{\wt}{wt}
\DeclareMathOperator{\aug}{aug}
\DeclareMathOperator{\adj}{adj}
\DeclareMathOperator{\Tr}{Tr}

\definecolor{tablehighlightblue}{RGB}{225,239,255}

\begin{document}
\title{Designing Group-Valued Codes with Full Regular Low-Weight Bases}
\author{Jong Yeon Lee}
\email{jongyeon@illinois.edu}
\affiliation{Department of Physics, Illinois Quantum Information Science and Technology Center, and Institute of Condensed Matter Theory, The Grainger College of Engineering, University of Illinois at Urbana-Champaign, Urbana, Illinois 61801, USA}
\affiliation{Korea Institute for Advanced Study, Seoul 02455, South Korea}
\begin{abstract}
Efficient fault-tolerant quantum computing architectures benefit from low-weight logical operators.
However, reported canonical bases for high-rate, high-distance codes can be several times heavier than the code distance, and fundamental obstructions can prevent minimum-weight logical operators from forming a complete canonical basis.
To address this issue, we develop design principles for group-valued quantum error-correcting codes with low-weight full regular canonical logical bases, generated by translating a set of seed operators.
We introduce frame width $f$, the smallest weight cutoff that permits a complete canonical logical basis, as an additional code parameter alongside distance $d$ and check weight $w$.
Cataloguing these parameters together with explicit logical bases can inform resource estimates and the compilation of fault-tolerant logical operations.

First, we identify a structural obstruction: for fully populated binary monomial CSS checks over a group of odd order, full regularity forces distance two, while for groups of power-of-two order, full row rank of both check matrices after replacing every group element by $1$ guarantees full regularity.
Second, we demonstrate our principles through explicit high-rate lifted-product (LP) and pair-partition (PP) codes: an $[[1088,128,22]]$ LP code with $f=22$ and an $[[1024,256,24]]$ PP code with $f \leq 27$, both at $w=10$.
At $w=11$, our PP constructions attain either $d=25$ with $f \leq 27$ or $d=24$ with the proved optimum $f=25$, demonstrating both a tradeoff among these parameters and a strict separation between distance and frame width.
Finally, we propose a method combining symmetry reduction and stabilizer-based pruning for exhaustive distance certification, and use it to efficiently establish exact distances for LP examples with more than a thousand qubits.
\end{abstract}
\maketitle

\section{Introduction}
The physical support of logical operators affects the resources required for fault-tolerant quantum computation.
The distance $d$ specifies the minimum weight of a nontrivial logical operator, whereas directly supplying canonical logical pairs for every encoded qubit requires a complete basis.
Reported canonical bases for high-rate codes can have maximum weights several times the available distance upper bounds, even after stabilizer-coset reduction~\cite[Table~I]{Lee2026Spectroscopy}.
We therefore ask how to design codes together with complete canonical bases whose operator weights can be controlled.

Logical Pauli measurements make this requirement concrete.
Surgery techniques for quantum low-density parity-check (LDPC) codes~\cite{CohenEtAl2022Surgery,CrossEtAl2024Surgery} and homomorphic or homological measurements~\cite{HuangEtAl2023Homomorphic,IdeEtAl2025Homological} couple ancillas to logical supports through ancilla constructions that preserve the required distance.
The weight of the measured operator affects the ancilla spacetime cost~\cite{WilliamsonYoder2026Gauging, SwaroopEtAl2026Adapters, YuanEtAl2026Parsimonious}.
Group symmetry offers a further design objective: in a full regular basis, all canonical pairs arise by translating a few seed pairs, allowing their coupling patterns to be reused~\cite{WebsterSmithCohen2025}.
We seek construction principles that guarantee this organization while controlling every basis operator's weight.

\input{main_v7_parts/figures/module_spectrum}

Group-algebra structure appears in various constructions, from generalized and bivariate bicycle codes to lifted product (LP) and pair-partition (PP) codes~\cite{KovalevPryadko2013,Bravyi2024,LinPryadko2024, PanteleevKalachev2021,PanteleevKalachev2022,BreuckmannEberhardt2021,LeeEtAl2026PP}.
Explicit logical bases support structured operations in hypergraph products~\cite{QuintavalleWebsterVasmer2023,XuEtAl2024} and bicycle codes~\cite{EberhardtSteffan2025,RabetiMahdavifar2026,Davenport2026}.
Logical spectroscopy constructs complete canonically paired logical bases for Abelian LP codes~\cite{Lee2026Spectroscopy}; its pairing framework also extends to PP codes~\cite{LeeEtAl2026PP}. 
Zheng et al. construct canonical logical bases for cyclic sectors of LP codes~\cite{ZhengEtAl2026CanonicalLP}.
Bhardwaj et al. and Hong construct non-Abelian LP codes with explicit logical bases generated from seed operators by group actions~\cite{BhardwajEtAl2026Processors,Hong2026ZSZ}.
These developments motivate a further design question: how can a code admit a complete canonical basis organized into regular orbits while keeping every basis operator's weight small, ideally close to distance?

\subsection{Summary of Results}

To quantify this objective, we use \emph{frame width} $f$, the smallest weight cutoff that permits a complete canonical logical basis.
The same idea was introduced for CSS bases as optimized conjugate-basis width~\cite[Sec.~IV.D]{Lee2026Spectroscopy}.
Together with distance $d$, check weight $w$, and explicit logical supports, it provides information relevant to logical measurements and fault-tolerant resource estimates.

Full regularity first imposes a structural constraint.
For Abelian groups of odd order, unequal logical multiplicities across symmetry components obstruct full regularity (Fig.~\ref{fig:spectrum}).
For fully populated binary monomial CSS checks at odd order, full regularity forces $d=2$ when $k>0$.
This obstruction is specific to the binary construction; nonbinary coefficients can avoid it.
For a group of power-of-two order, full row rank of both checks after replacing group elements by $1$ and summing coefficients guarantees full regularity (Theorem~\ref{thm:free}).
We can therefore choose zero patterns and coefficients that ensure complete regular orbits before searching the group labels, subject to CSS commutation.

Our design strategy first sets the encoding rate $k/n$ through the matrix dimensions: those of the classical constituents for LP codes and of the group-valued check matrices for PP codes.
Under the corresponding full-rank augmentation conditions, these dimensions fix the rate and guarantee full regularity over a finite $2$-group.
For a fixed coefficient field, the group order then sets the blocklength.
Within the resulting commuting families, we seek large distance under prescribed limits on check weight and reported frame width $f_{\mathrm{rep}}$.

In general, sparse checks do not imply sparse logical operators: solving the check equations and imposing canonical pairing can make logical operators much heavier.
We choose families that control both steps.
For our explicit LP and PP frames, we first specify the $X$- and $Z$-type seeds on a common qubit block.
We then choose their components on disjoint sets of other blocks so that both operators commute with every stabilizer while maintaining their shared support.
Sparse corrections keep the seeds low-weight, while their disjoint supports ensure that only the common block contributes to the pairings between translates.
The LP arrangement gives canonical pairing directly; for PP codes, we identify additional constraints that allow canonical pairing without increasing the seed weights (Fig.~\ref{fig:unified-frames}).

In addition, quaternary coefficients enlarge the useful design space~\cite{Kasai_2012,panteleev2022asymptoticallygoodquantumlocally,OkadaKasai2026LWCB}.
They allow fully populated matrices to satisfy the rank conditions for full regularity, while their additional cancellation patterns help keep the logical corrections sparse.
These choices must also control the physical weights, since quaternary expansion can increase binary check weight.
For each candidate, we therefore evaluate the binary check weights and the physical weights of the normalized seeds.
The largest seed weight is the reported frame width, because translation preserves weight, and supplies an upper bound on $f$.
Only candidates satisfying the prescribed weight limits are retained.


\begin{table}[!t]
\centering\small\setlength{\tabcolsep}{4pt}
\begin{tabular}{@{}cclrrr@{}}
\toprule
Family & Field & $[[n,k,d]]$ & $w$ & $w_c$ & $f_{\mathrm{rep}}$ \\
\midrule
\multirow{9}{*}{\shortstack{LP\\(rate $2/17$)}} & $\mathbb F_4$ & $[[476,56,14]]$ & $10$ & $6$ & $20$ \\
 & $\mathbb F_4$ & $[[544,64,14]]$ & $8$ & $5$ & $22$ \\
 & $\mathbb F_4$ & $[[544,64,14]]$ & $9$ & $6$ & $16$ \\
 & $\mathbb F_4$ & $[[544,64,15]]$ & $10$ & $6$ & $22$ \\
 & $\mathbb F_2$ & $[[1088,128,12]]$ & $8$ & $5$ & $^\star 12$ \\
 & $\mathbb F_4$ & $[[1088,128,16]]$ & $8$ & $5$ & $26$ \\
 & $\mathbb F_4$ & $[[1088,128,16]]$ & $9$ & $5$ & $18$ \\
 & $\mathbb F_4$ & $[[1088,128,20]]$ & $10$ & $6$ & $22$ \\
 & $\mathbb F_4$ & $[[1088,128,22]]$ & $10$ & $6$ & $^\star 22$ \\
\midrule
\multirow{15}{*}{\shortstack{PP\\(rate $1/4$)}} & $\mathbb F_4$ & $[[256,64,12]]$ & $10$ & $5$ & $19$ \\
 & $\mathbb F_4$ & $[[512,128,11]]$ & $8$ & $4$ & $16$ \\
 & $\mathbb F_4$ & $[[512,128,19]]$ & $9$ & $4$ & $30$ \\
 & $\mathbb F_4$ & $[[512,128,19]]$ & $10$ & $4$ & $25$ \\
 & $\mathbb F_4$ & $[[512,128,19]]$ & $11$ & $5$ & $24$ \\
 & $\mathbb F_4$ & $[[512,128,20]]$ & $10$ & $5$ & $31$ \\
 & $\mathbb F_4$ & $[[512,128,20]]$ & $11$ & $5$ & $28$ \\
 & $\mathbb F_4$ & $[[768,192,13]]$ & $8$ & $4$ & $20$ \\
 & $\mathbb F_4$ & $[[768,192,20]]$ & $9$ & $4$ & $30$ \\
 & $\mathbb F_4$ & $[[768,192,22]]$ & $10$ & $5$ & $29$ \\
 & $\mathbb F_4$ & $[[1024,256,12]]$ & $8$ & $4$ & $15$ \\
 & $\mathbb F_4$ & $[[1024,256,18]]$ & $9$ & $4$ & $26$ \\
 & $\mathbb F_4$ & $[[1024,256,24]]$ & $10$ & $4$ & $27$ \\
 & $\mathbb F_4$ & $[[1024,256,24]]$ & $11$ & $5$ & $^\star 25$ \\
 & $\mathbb F_4$ & $[[1024,256,25]]$ & $11$ & $5$ & $27$ \\
\bottomrule
\end{tabular}
\caption{ {\bf Selected codes.} We present LP and PP codes with full regular canonical frames, ordered within each family by length, reported distance, and check weight.
All listed distances are exact.
$w$ and $w_c$ are the maximum check and column weights per CSS sector; $f_{\mathrm{rep}}$ is the reported frame width and bounds $f$ from above.
A star marks $f_{\mathrm{rep}}=f$.
LP periods are $32$ over $\mathbb F_2$ and $7$, $8$, or $16$ over $\mathbb F_4$; PP periods are $16$, $32$, $48$, and $64$ at lengths $256$, $512$, $768$, and $1024$.
For each quaternary LP check, we choose an independent pair of binary rows that minimizes their maximum weight.
Construction details appear in Appendix~\ref{sec:construction-catalogue}; all inputs and frames are supplied in Ref.~\cite{Lee2026RegularFramesRepository}.}
\label{tab:code-comparison}
\label{tab:lp-codes}
\label{tab:quaternary}
\end{table}

We vary the remaining group labels and coefficients to search for codes with larger distance.
Building on logical spectroscopy's distance filters~\cite{Lee2026Spectroscopy}, searches on smaller constituent and quotient codes reject candidates whose distance upper bounds fall below the target.
For the tabulated codes, exhaustive full-code searches exclude every logical operator lighter than a known nontrivial witness.
The construction supplies an explicit canonical frame; the search selects codes with larger distance, which exhaustive exclusions then certify.

Our constructions give explicit examples combining substantial distance with low-weight complete canonical bases.
At check weight $w=10$, we obtain an $[[1088,128,22]]$ LP code with proved optimal frame width $22$ and an $[[1024,256,24]]$ PP code with reported frame width $27$.
Within the same PP length and rate, allowing $w=11$ yields either distance $25$ with reported frame width $27$, or distance $24$ with the proved optimum $f=25$.
These examples illustrate how check weight can be exchanged for distance or for a narrower canonical frame.
The distance-$24$, width-$25$ PP example also establishes a fundamental distinction between finding a shortest logical operator and supplying a complete logical basis.
Although weight-$24$ logical operators exist, even the collection of all such operators fails to span the logical space.
Every complete canonical basis must therefore contain an operator of weight at least $25$.
Our full regular frame attains this bound, proving $f=25$. 

The distance-$20$ LP example exhibits a different limitation.
Each logical sector separately has a complete regular basis consisting of weight-$20$ operators.
Nevertheless, these minimum-weight operators cannot be arranged into a full regular canonical frame under the prescribed $C_{16}$ action: the required pairings cannot all be satisfied simultaneously.
Every such frame has width at least $21$, while our construction achieves $22$.
Here the obstruction concerns canonical pairing compatible with the specified symmetry; the unrestricted frame width remains unresolved.
Thus low-weight spanning and low-weight canonical pairing are separate design requirements, even when full regularity is already guaranteed.

Sections~\ref{sec:obstruction}--\ref{sec:modular} establish the existence and seed-selection criteria; Secs.~\ref{sec:lp}--\ref{sec:binary} apply them to LP and PP codes.
Appendix~\ref{sec:distance-algorithm} describes the exhaustive searches and physical witnesses establishing the tabulated distance bounds.

\section{Symmetry obstruction}
\label{sec:obstruction}
\subsection{From translated operators to a complete basis}
Suppose the physical qubits are arranged in $L$ blocks of $P=|G|$ qubits, with the qubits in each block labeled by the elements of $G$.
The group action permutes these $P$ labels within every block.
Starting from $q$ pairs of seed operators therefore produces $qP$ candidate logical pairs.
Full regularity means that these translates are independent modulo stabilizers, canonically paired, and account for every logical qubit.
The algebra below tests precisely these three requirements.

Set $\R=\F[G]$.
We use columns of free right $\R$-modules, with matrices acting by left multiplication.
Simultaneous right multiplication by a group element therefore permutes the physical qubits and preserves the check spaces.
For $a=\sum_g a_g g$, write $\bar a=\sum_g a_g g^{-1}$ and $\aug(a)=\sum_g a_g$.
Set $M^\dagger=\overline M^{\mathsf T}$.
The involution reverses multiplication order, so these conventions also apply to non-Abelian groups.

Let $H_X\in\R^{r_X\times L}$ and $H_Z\in\R^{r_Z\times L}$ satisfy $H_XH_Z^\dagger=0$.
Write $\widehat H_X,\widehat H_Z$ for their binary regular-representation expansions, whose rows specify the $X$- and $Z$-type stabilizer generators.
We use $[[n,k,d]]$ for a code on $n$ physical qubits encoding $k$ logical qubits, with $d$ the minimum weight of a nontrivial logical Pauli operator.
Here $n=LP$ and
\begin{align}
k=LP-\rank_{\F}\widehat H_X-\rank_{\F}\widehat H_Z. \label{eq:binary-dimension}
\end{align}
For the Calderbank--Shor--Steane (CSS) construction~\cite{CalderbankShor1996,Steane1996}, the logical quotients are
\begin{align}
\LX&=\ker H_Z/\im H_X^\dagger, \\
\LZ&=\ker H_X/\im H_Z^\dagger. \label{eq:logicals}
\end{align}
These quotients identify two operators when they differ by a stabilizer.

Full regularity is defined relative to the chosen right $G$-action; the binary CSS matrices alone do not specify this action.
For a binary group-algebra vector $v$, $\wt(v)$ counts its nonzero coefficients in the physical group-element basis.
When $k>0$, the $Z$-distance $d_Z$ is the minimum weight in $\ker H_X\setminus\im H_Z^\dagger$; define $d_X$ by exchanging the sectors, and set $d=\min(d_X,d_Z)$.

The binary logical pairing $B(x,z)=\widehat x^{\mathsf T}\widehat z$, where $\widehat x$ and $\widehat z$ are the binary coefficient vectors, is the nondegenerate CSS commutation pairing~\cite{Gottesman1997}.
Equivalently, $B(x,z)$ is the coefficient of $1_G$ in $x^\dagger z$.
The group-algebra pairing $x^\dagger z$ records the pairings with every translate at once.
It descends to logical classes because $H_Xz=H_Zx=0$, and satisfies $(xa)^\dagger(zb)=\bar a(x^\dagger z)b$.
In particular, its first argument is conjugate-linear rather than linear over a noncommutative group algebra.

\begin{samepage}
\begin{definition}
\label{def:frame}
A \emph{full regular canonical frame} for a code with $k=qP$ consists of seed columns $X=(x_1,\ldots,x_q)$ and $Z=(z_1,\ldots,z_q)$ satisfying
\begin{align}
H_ZX=0,\qquad H_XZ=0,\qquad X^\dagger Z=I_q. \label{eq:frame}
\end{align}
The translates of these seeds form a \emph{full regular logical basis}, with canonical $X$--$Z$ pairing.
We use \emph{frame} for the seed data and \emph{basis} for their translates.
\end{definition}
\end{samepage}
Indeed, $B(x_i g,z_j h)=\delta_{ij}\delta_{g,h}$ for $g,h\in G$.
The word canonical refers to the pairing, not to a unique choice of basis.

A copy of $\R$ describes one logical seed together with all its independent translates.
The logical module is \emph{free} when it is a direct sum of such copies.

\begin{proposition}
\label{prop:module}
A full regular canonical frame exists if and only if $\LZ\cong\R^q$ for some $q$.
Then $\LX\cong\R^q$ as well.
\end{proposition}
\begin{proof}
The translates of an $\R$-basis are a binary regular basis, and conversely a binary regular basis specifies an $\R$-basis.
Choose such a basis of $\LZ$ and its binary dual in $\LX$.
Invariance of the pairing under simultaneous translation makes the dual basis transform by the same regular permutation action.
Choose physical representatives for the seeds of both bases.
They satisfy Eq.~\eqref{eq:frame}, which records their pairings with all translates.
\end{proof}

\subsection{Weight of a complete basis}
For a stabilizer code $\mathcal C$ with $k>0$, define its \emph{frame width} by
\begin{align}
f(\mathcal C)=\min_{\mathcal B}\max_{O\in\mathcal B}\wt(O), \label{eq:frame-parameter}
\end{align}
where $\mathcal B$ ranges over physical representatives of complete canonical logical Pauli bases and $\wt(O)$ counts their occupied qubits.
The $2k$ operators form $k$ anticommuting pairs and commute between distinct pairs; mixed Pauli representatives are allowed, and no group action is specified.

For a CSS code, define the \emph{spanning threshold} in the $Z$ sector by
\begin{align}
f_{\mathrm{span}}^Z=\min\bigl\{t:\operatorname{span}_{\F}\{[z]:\wt(z)\le t\}=\LZ\bigr\}.
\label{eq:spanning-threshold}
\end{align}
Here $z$ ranges over physical binary representatives in $\ker\widehat H_X$; quaternary constructions are evaluated after binary expansion.
Define $f_{\mathrm{span}}^X$ by exchanging the sectors and set $f_{\mathrm{span}}=\max\{f_{\mathrm{span}}^X,f_{\mathrm{span}}^Z\}$.
This is the weight cutoff needed to span both sectors, without requiring canonical pairing or regular orbits.

For specified seeds $(X,Z)$ of a full regular CSS frame, the \emph{reported frame width} is $f_{\mathrm{rep}}=\max_i\{\wt(x_i),\wt(z_i)\}$. This quantity depends on the chosen frame.
Define $f_{\rm reg}(\mathcal C,G)$ by minimizing this width over all full regular canonical CSS frames for the specified group action, and set it to $\infty$ if none exists.
For a CSS code with an explicit full regular frame,
\begin{align}
d\le f_{\mathrm{span}}\le f\le f_{\rm reg}\le f_{\mathrm{rep}}. \label{eq:frame-width-hierarchy}
\end{align}
Projecting a complete Pauli basis onto either CSS sector gives a spanning set without increasing any operator's weight, proving $f_{\mathrm{span}}\le f$ even for mixed representatives.
The remaining inequalities follow from the definitions and invariance of weight under translation.
We suppress the code and group arguments when they are fixed.
Tables report $f_{\mathrm{rep}}$, with a left superscript $\star$ precisely when equality to the optimized parameter $f$ has been proved.

The spanning threshold also has a dual characterization.
For a nonzero $\F$-linear functional $\lambda\in\LZ^*$, minimizing over the same physical representatives gives
\begin{align}
d_\lambda&=\min\{\wt(z):\lambda([z])=1\},\notag\\
f_{\mathrm{span}}^Z&=\max_{0\ne\lambda\in\LZ^*}d_\lambda.
\label{eq:spanning-dual}
\end{align}
The second equality holds because the classes represented at weight at most $t$ fail to span precisely when some nonzero functional evaluates to zero on every such class.
Thus excluding representatives with $\lambda([z])=1$ through a given weight certifies a lower bound on $f_{\mathrm{span}}^Z$, and hence on $f$.

\subsection{The odd-order obstruction}
For an odd-order Abelian group, the translation action separates into independent symmetry components.
A regular orbit contributes once to each component, with its field degree accounting for the binary dimension.
A complete collection of regular orbits must therefore have the same multiplicity in every component, as illustrated in Fig.~\ref{fig:spectrum}.
Algebraically, $\R$ is a product of finite fields $K_\nu$.
Projection to each factor preserves kernels and images, so
\begin{align}
h_\nu:=\dim_{K_\nu}(\LZ)_\nu&=L-\rank_{K_\nu}(H_X)_\nu \notag\\
&\quad-\rank_{K_\nu}(H_Z^\dagger)_\nu. \label{eq:component-ranks}
\end{align}
A module over this product is free of rank $q$ exactly when every component has dimension $q$.
Thus a full regular canonical frame exists exactly when $h_\nu$ is independent of $\nu$.
For cyclic $G$, the factors come from the distinct irreducible factors of $t^P-1$.
If both checks are fully populated with monomials, evaluation at $t=1$ makes each an all-one matrix of rank one, giving $h_1=L-2$.
A component in which the syndrome and stabilizer maps both have rank $J>1$ instead has $h_\nu=L-2J$, obstructing freeness.
An invariant vector is not itself an obstruction: the sum of a regular orbit is invariant.
The obstruction is the mismatch of multiplicities.

Full regularity imposes a stronger restriction on fully populated binary monomial checks, including over non-Abelian groups.
When $k>0$, it forces each check space to be generated by a single block row, allowing two qubits to cancel their syndromes and form a logical operator.
The proof uses Maschke's averaging argument, which makes taking invariants exact at odd group order~\cite{Iyengar2004}.
\begin{samepage}
\begin{proposition}[Odd order]
\label{prop:odd}
Let $G$ have odd order, let $r_X,r_Z\ge1$, and let every entry of both CSS checks be a group element.
If $k>0$ and the code has a full regular canonical frame, then $d=2$.
\end{proposition}
\end{samepage}
\begin{proof}
Averaging over $G$ is a projection onto invariant vectors because $P=1$ in $\F$.
Taking invariants is exact: an invariant image has an invariant preimage obtained by averaging any preimage.
The invariant part of a regular module is one-dimensional, spanned by $\sum_g g$.
On invariant vectors each check acts by its augmentation, an all-one matrix of rank one.
Thus $\dim_{\F}(\LZ)^G=L-2$.
If $\LZ\cong\R^q$, this dimension equals $q$, so $k=P(L-2)$.

The two expanded check ranks consequently sum to $2P$.
Each is at least $P$, since one fully populated monomial block row already has rank $P$.
Both ranks equal $P$, so each check space is generated by one such block row.
Write the generating $X$ block row as $(h_1,\ldots,h_L)$.
For distinct $\ell,\ell'$, set $z_\ell=h_\ell^{-1}$, $z_{\ell'}=h_{\ell'}^{-1}$, and all other entries of $z$ to zero.
This weight-two vector has syndrome $1_G+1_G=0$ against that block row, and hence $H_Xz=0$.
In this one-block-row presentation, every nonzero $Z$ stabilizer occupies all $L$ blocks because each monomial block is invertible.
Since $k>0$ implies $L>2$, the weight-two operator is not a stabilizer.
A weight-one operator has nonzero syndrome because every block entry is invertible.
Hence $d=2$.
\end{proof}
Here fully populated refers to the CSS checks $H_X,H_Z$ themselves; product-code checks contain structural zeros even when their constituents are fully populated.

Over $\mathbb F_4$, augmentation sends $c g$ to $c$, so a fully populated monomial check need not collapse to rank one.
The binary proof therefore does not extend directly.
Our $C_7$ quaternary LP examples show that odd group order alone does not preclude full regularity at distance greater than two (Appendix~\ref{sec:lp-short-quaternary}).
Their product checks also contain structural zeros, so these examples do not isolate the effect of changing the coefficient field.
At odd order, however, augmentation alone does not certify full regularity; the other symmetry components must also be controlled, for example through an explicitly invertible pivot.

\section{Full regularity over \texorpdfstring{$2$}{2}-groups}
\label{sec:modular}
For a finite $2$-group, a sufficient test for full regularity uses only the augmented checks, which forget the group labels and retain their coefficient sums.

The reason differs from the odd-order case: in characteristic two, an even-order action need not split into separate symmetry components.
For the cyclic group of order $P=2^a$,
\begin{align}
\F[C_P]\cong\F[\epsilon]/(\epsilon^P),\qquad \epsilon=t+1. \label{eq:localcyclic}
\end{align}
The invariant vectors now lie inside translation chains linked by $t+1$, rather than in a separate field component.
If $T$ denotes translation by $t$, then $\epsilon$ acts as the nilpotent operator $T+I$.
A full regular module consists of chains $z_i,z_i\epsilon,\ldots,z_i\epsilon^{P-1}$ of length $P$; the last vector in each chain is invariant under translation.
The vectors in each chain are obtained by an invertible change of basis from the $P$ translates of one seed.
The augmentation criterion below ensures that these chains have the full length $P$, replacing the separate rank tests of the odd-order case.

For a finite $2$-group, a square matrix over $\R$ is invertible whenever its augmentation is invertible.
The underlying algebra is local: the augmentation ideal $J=\ker\aug$ is nilpotent and $\R/J\cong\F$.
Every finitely generated projective module over this algebra is free~\cite{Iyengar2004}.
Applied to the CSS check maps, these standard algebraic facts give a coding criterion for full regularity, valid also for non-Abelian groups.

\begin{samepage}
\begin{theorem}[Full-rank augmentation]
\label{thm:free}
Let $G$ be a finite $2$-group and $H_XH_Z^\dagger=0$.
If $B_X=\aug(H_X)$ and $B_Z=\aug(H_Z)$ have full row ranks $r_X,r_Z$, then the expanded checks have ranks $r_XP,r_ZP$, and
\begin{align}
\LX\cong\LZ\cong\R^q,\qquad k=qP, \label{eq:free}
\end{align}
where $q=L-r_X-r_Z$.
For $q>0$, the code has a full regular canonical frame.
\end{theorem}
\end{samepage}
\begin{proof}
Choose $r_X$ columns whose augmentation is an invertible binary matrix $U_0$.
The corresponding square block is $U=U_0(I+N)$ with $N\in M_{r_X}(J)$.
If $J^s=0$, then $N^s=0$ and $(I+N)^{-1}=\sum_{j=0}^{s-1}N^j$.
Thus $H_X$ is split surjective, and $\ker H_X$ is projective.
Likewise $H_Z$ is split surjective, so $H_Z^\dagger$ is split injective.
Its image lies in $\ker H_X$, and its left inverse restricts to that kernel, splitting the quotient defining $\LZ$.
The two splittings give $\R^L\cong\R^{r_X}\oplus\R^{r_Z}\oplus\LZ$ noncanonically, so $\LZ$ is projective and hence free over the local algebra.
Subtracting ranks gives $q=L-r_X-r_Z$.
Exchanging $X$ and $Z$ proves the other assertion, and Proposition~\ref{prop:module} supplies the frame.
\end{proof}

For monomial-or-zero checks, augmentation gives the binary support, so a full-rank support guarantees full regularity for every commuting lift over a finite $2$-group.
This lets us fix the support and coefficients to guarantee full regularity before optimizing the commuting group labels, without changing $k$.
The explicit low-weight LP and PP constructions below use Abelian groups.
The condition is sufficient; a presentation with redundant check rows can also have a free logical module.

Even cyclic order alone does not make augmentation sufficient: for periods other than powers of two, the same splitting argument requires full rank in every residue-field component, of which augmentation tests only $t=1$ (Appendix~\ref{sec:other-periods}).

\subsection{Selecting seeds for full regular orbits}
\label{sec:seed-selection}
The augmented checks also determine which candidate logical operators can serve as seeds of a full regular basis.
For a binary operator, augmentation records the parity of its support within each group block.
Under Theorem~\ref{thm:free}, candidate seeds generate a complete regular basis precisely when these block-parity vectors form a basis modulo the augmented stabilizers.
Algebraically, reducing the split sequences modulo $J$ gives
\begin{align}
\LZ/\LZ J\cong\ker B_X/\im B_Z^{\mathsf T}. \label{eq:residue}
\end{align}
Classes $[z_1],\ldots,[z_q]\in\LZ$ form an $\R$-basis exactly when their residues form a binary basis of this quotient.
To see sufficiency, express them in any $\R$-basis of the free module $\LZ$.
Their coordinate matrix has invertible augmentation, so the finite inverse in the proof of Theorem~\ref{thm:free} makes it invertible over $\R$.
Necessity follows by reducing a basis modulo $J$.

In particular, a logical operator can be included among the seeds of a full regular basis exactly when its residue is nonzero.
A minimum-weight logical operator can fail this test: a nontrivial logical class need not generate an orbit that extends to a full regular basis.
The seed search should therefore select low-weight representatives whose residues form a basis of the $q$-dimensional binary quotient in Eq.~\eqref{eq:residue}.
Their $qP$ translates then span the entire $Z$-type logical space.
The remaining task is to find low-weight $X$-type partners with canonical pairing.

\subsection{Controlling the weight of canonical partners}
\label{sec:light-partners}
Complete regular bases in the two sectors need not yet be canonically paired.
Their pairing matrix specifies how to recombine the seeds to obtain canonical partners.
Under Theorem~\ref{thm:free}, collect $q$ kernel representatives from each sector into $X_0$ and $Z_0$, and require their residues to form bases of the two augmented logical quotients.
The matrix $M=X_0^\dagger Z_0$ records a commutation table: its coefficients give the pairings of each seed with all translates of the opposite-sector seeds.
Its augmentation is the nondegenerate binary pairing between the residue bases, so $M$ is invertible.
Consequently
\begin{align}
X=X_0,\qquad Z=Z_0M^{-1},\qquad X^\dagger Z=I_q. \label{eq:pairing-normalization}
\end{align}
Logical spectroscopy uses the same Gram-matrix normalization for complete conjugate bases between reciprocal primary components, including at even lift sizes~\cite[Appendix~D, Lemma~12]{Lee2026Spectroscopy}.
For regular orbits, the corresponding condition is identity correlation in the full primary rings~\cite[Appendix~D.3]{LeeEtAl2026PP}.
Augmentation certifies invertibility; Eq.~\eqref{eq:pairing-normalization} uses the full inverse over $\R$, including its nilpotent terms.
Because this basis change acts on seeds over $\R$, their translates remain organized into regular orbits.
The identity pairing proves independence modulo stabilizers, and $qP=k$ proves completeness.
The same normalization applies whenever $M$ is invertible, including over non-Abelian group algebras.

Applying $M^{-1}$ can combine many translated seeds and increase their physical support.
Permuting the seeds and translating each by a group element preserves these weights.
The LP and cofactor PP constructions below control both the weight of syndrome corrections and the effect of normalizing the pairing.
They place the corrections for $X$- and $Z$-type seeds on disjoint sets of coordinates, so the corrections do not contribute to the pairing.
Only the physical coordinate blocks on which both seeds have support contribute, as illustrated in Fig.~\ref{fig:unified-frames}.
For LP codes, the tensor arrangement gives identity pairing.
For PP codes, disjoint pivot supports give a scalar pairing; when its inverse is a group translation, normalization preserves the seed weights.
\input{main_v7_parts/figures/unified_frame_construction}

\section{Lifted-product codes}
\label{sec:lp}
For lifted-product (LP) codes, the two design choices can be made at the constituent level: ranks determine full regularity, and sparse corrections in the constituent kernels control the logical-seed weights.
Hypergraph products and their lifted versions realize this structure~\cite{TillichZemor2014,PanteleevKalachev2022}, but a systematic hypergraph-product basis need not retain canonical pairings compatible with the lift~\cite{Scruby2025}.
For Abelian $G$, take $A\in\R^{m_a\times n_a}$ and $B\in\R^{m_b\times n_b}$ and form
\begin{align}
H_X&=\bigl(A\otimes I_{m_b}\quad I_{m_a}\otimes B\bigr), \\
H_Z&=\bigl(I_{n_a}\otimes B^\dagger\quad A^\dagger\otimes I_{n_b}\bigr). \label{eq:lp}
\end{align}
Commutativity of $\R$ gives $H_XH_Z^\dagger=0$, with $L=n_am_b+m_an_b$ blocks of $P$ qubits.

\subsection{Full regularity of lifted products}
For an Abelian $2$-group, full row rank of $\aug(A)$ and full column rank of $\aug(B)$ give full row rank of the augmented LP checks.
Theorem~\ref{thm:free} then guarantees $\LZ\cong\LX\cong\R^q$, where $q=(n_a-m_a)(m_b-n_b)$.
At odd order, full row rank of $A$ and full column rank of $B$ in every field component give the corresponding guarantee~\cite[Sec.~IV.3]{Lee2026Spectroscopy}; Appendix~\ref{sec:other-periods} gives the general component criterion.
For $\mathsf M^{45}_{5,9}$, the constituent rank defect instead leaves $720$ logical pairs in regular orbits and $64$ additional translation-invariant pairs (Fig.~\ref{fig:spectrum})~\cite[Sec.~V]{Lee2026Spectroscopy}.

\subsection{Explicit low-weight frames}
\label{sec:lp-explicit-frames}
We use the self-adjoint family $\mathrm{LP}(A,A^\dagger)$ over an Abelian $2$-group, taking $A\in\R^{r\times s}$ with $r<s$ and full-row-rank augmentation.
Its parameters are
\begin{align}
n=P(s^2+r^2),\qquad k=P(s-r)^2. \label{eq:lp-free-parameters}
\end{align}
Thus $r,s$ fix the rate $(s-r)^2/(s^2+r^2)$, while $P$ sets the length.
The remaining design task is to keep the constituent-kernel corrections sparse while searching for distance.

Choose $r$ columns whose augmented submatrix is invertible and order them first, $A=(D\ E)$.
We call the columns of $D$ pivot coordinates and those of $E$ free coordinates.
The argument in the proof of Theorem~\ref{thm:free} makes $D$ invertible.
Setting one free entry to $1$ and the others to zero then determines the pivot entries of a kernel vector as the corresponding column of $F=D^{-1}E$.
Arrange the first $s^2$ qubit blocks as an $s\times s$ array and set the remaining $r^2$ blocks to zero.
Then $H_X$ applies $A$ to each column, while $H_Z$ applies $A$ to each row, viewed as a column vector; both directions use $A$ because $B=A^\dagger$.
Place the two seeds' corrections on pivot--free and free--pivot coordinates, leaving only their common free--free coordinate shared [Fig.~\ref{fig:unified-frames}(a)].
Set $U=\binom{F}{I_{s-r}}$ and $V=\binom{0}{I_{s-r}}$, so $U$ includes the correction and $V$ selects the free coordinate.
With columns $u_i,v_i$, take
\begin{align}
z_{ij}=(u_i\otimes v_j,0),\qquad x_{ij}=(v_i\otimes u_j,0). \label{eq:lp-regular-seeds}
\end{align}
Here the augmentation criterion selects an invertible pivot, and the weights of $D^{-1}E$ guide the choice among such pivots.
Each occupied column or row is in $\ker A$ because $AU=0$, so these seeds have zero syndrome.
The identities $V^\dagger U=U^\dagger V=I$ give canonical pairing for all seed labels (Fig.~\ref{fig:unified-frames}).
The $(s-r)^2$ seed pairs therefore give a full regular frame with $f_{\mathrm{rep}}=\max_i\wt(u_i)$.

For binary cyclic lifts, a monomial pivot determinant makes $D^{-1}$ a translate of the cofactor matrix; cancellation in $D^{-1}E$ can reduce the correction weights further.
The binary example in Table~\ref{tab:code-comparison} attains $f=d=12$ this way.
Three-row binary monomial-or-zero constituents in this full-rank construction obey $d\le14$, motivating polynomial entries, more rows, or quaternary coefficients (Appendix~\ref{sec:lp-design-details}).

Quaternary coefficients~\cite{Kasai_2012, panteleev2022asymptoticallygoodquantumlocally, OkadaKasai2026LWCB} are a shared design choice for LP and PP codes: they allow fully populated pivots with full-rank augmentation and additional determinant cancellations.
Write $\mathbb F_4=\{0,1,\omega,\omega^2\}$ with $1+\omega+\omega^2=0$; three determinant terms with coefficients $1$, $\omega$, and $\omega^2$ and the same group label therefore cancel.
For LP codes, such cancellations can leave a monomial pivot determinant, allowing the weights of $D^{-1}E$ to be controlled through cofactors; Appendix~\ref{sec:lp-quaternary-design} gives an explicit pivot and its determinant.
The augmentation and residue arguments apply over $\mathbb F_4[G]$, with ranks and residue bases now computed over $\mathbb F_4$.

To obtain a qubit CSS code, we expand each field coordinate into two binary coordinates in the self-dual trace basis $(\omega,\omega^2)$.
This preserves commutation and canonical pairing; Appendix~\ref{sec:binary-expansion} gives the expansion and physical-weight formulas.
Multiplication by $\omega$ or $\omega^2$ can turn a one-bit coefficient into a two-bit coefficient, so quaternary symbol weight does not determine physical weight.
All weights below count the resulting physical-qubit support.
This expansion gives
\begin{align}
n=2P(s^2+r^2),\qquad k=2P(s-r)^2, \label{eq:lp-f4-parameters}
\end{align}
with two binary seed pairs for each pair in Eq.~\eqref{eq:lp-regular-seeds}.

\subsection{Distance, check weight, and frame width}
The LP block of Table~\ref{tab:code-comparison} compares distance, check weight, and frame width at rate $2/17$.
At $(n,k)=(1088,128)$, quaternary coefficients increase distance beyond the binary example's $d=f=12$: at check weight $10$, both the distance-$20$ and distance-$22$ codes have reported width $22$, and the latter attains $f=f_{\mathrm{reg}}=d=22$.
Keeping most coefficients equal to $1$ limits the binary check support, since each acts as the identity on the two trace-basis coordinates; sparse corrections $D^{-1}E$ control the frame width.
Both weights must be evaluated for each choice of coefficients and group labels.

For each quaternary check, we choose two of its three nonzero binary trace combinations to minimize the larger row weight.
For fully populated monomial $3\times5$ quaternary constituents in the self-adjoint family with full-rank augmentation, $w\ge10$ in both the standard trace-basis expansion and this choice of generators (Appendix~\ref{sec:lp-check-floor}).
The bound is attained by the distance-$20$ and distance-$22$ examples.
The construction with one zero in each nonpivot column achieves $w=9$ and $d=16$, compared with $d=20$ or $22$ for the reported $w=10$ examples.
A different pattern, with two zero entries in the pivot, attains check weight $8$ at distances $14$ and $16$ for lengths $544$ and $1088$, respectively (Appendix~\ref{sec:lp-eight}).
Appendix~\ref{sec:lp-inputs} describes the shorter codes, general unit pivots, and independent-constituent variants.

Searches on smaller constituent and quotient-lift kernels reject candidates below the target distance early~\cite[App.~F]{Lee2026Spectroscopy}, but constituent distance alone does not certify quantum distance (Appendix~\ref{sec:lp-metric}).
In the full-code search, translation symmetry removes equivalent starting positions, while stabilizer-overlap constraints discard branches that cannot contain a minimum-weight logical: such a representative cannot be shortened by adding a stabilizer.
Completed exclusions through weight $D-1$ in both CSS sectors, together with a verified weight-$D$ logical witness, establish $d=D$.
Appendix~\ref{sec:distance-algorithm} gives the branching, pruning, and completeness arguments, including how a verified sector-exchange symmetry transfers exclusions between sectors.

\subsection{Spanning and canonical pairing under a fixed action}
\label{sec:lp-pairing-main}
Minimum-weight regular bases can span both logical sectors without admitting a canonical pairing of the same weight under the prescribed cyclic action.
Denote the $[[1088,128,16]]$ code with $w=9$ by $\mathcal C^{\mathrm{LP}}_{16,9}$ and the $[[1088,128,20]]$ code with $w=10$ by $\mathcal C^{\mathrm{LP}}_{20,10}$; the subscripts specify $(d,w)$.
Both have regular bases of weight $d$ in each sector separately, so $f_{\mathrm{span}}^X=f_{\mathrm{span}}^Z=d$.
For the specified $C_{16}$ action, however, their canonical regular widths satisfy $f_{\rm reg}=18$ and $21\le f_{\rm reg}\le22$, respectively (Proposition~\ref{prop:lp-regular-pairing-cost}).

The obstruction lies in choosing all canonical pairs simultaneously.
For a candidate $X$ seed and $Z$ seed, exactly one translate of the $Z$ seed must anticommute with the $X$ seed; every other translate must commute.
Two such pairs can belong to the same canonical frame only if their cross-pairings vanish at every translation.
Exhaustively enumerated lists of low-weight operators give a finite graph whose vertices are candidate orbit pairs and whose edges join compatible pairs.
A full frame needs $k/P=128/16=8$ mutually compatible pairs.
For lists complete through weights $17$ and $20$, respectively, proper colorings with five and three colors exclude such a choice: mutually adjacent vertices must all have different colors.

These exclusions concern the prescribed orbit organization.
The unrestricted width $f$ remains unresolved.
For the distance-$16$ example, any width-$17$ canonical basis could be reduced to width $16$ by multiplying its operators by stabilizers; hence the remaining possibilities are $f=16$ and $f=18$.
Thus $f_{\mathrm{span}}<f_{\rm reg}$ is proved for these examples, but $f_{\mathrm{span}}<f$ is not.
This mechanism differs from the PP example below, where low-weight logicals fail to span the logical space at all.
Appendix~\ref{sec:lp-pairing-obstruction} proves these bounds, including for regular symplectic frames with mixed Pauli seeds.

\section{Pair-partition codes}
\label{sec:binary}
In the binary PP construction, paired columns contribute identical group-algebra terms to each commutation equation.
These terms cancel in characteristic two.
We combine this freedom with full-rank supports and then use quaternary coefficients to control the weights of canonical partners.

\subsection{Pair cancellation and reflected partners}
For one row of $H_X$ and one row of $H_Z$, each column occupied by both contributes a group element to the corresponding entry of $H_XH_Z^\dagger$.
Pairing columns with equal contributions makes their sum zero in characteristic two.
Choose full-row-rank binary support matrices $B_X,B_Z$ with $B_XB_Z^{\mathsf T}=0$.
For each row pair $(i,j)$, partition the columns occupied by both rows into pairs $\{u,v\}$ and impose
\begin{align}
a_{iu}b_{ju}^{-1}=a_{iv}b_{jv}^{-1}. \label{eq:pp}
\end{align}
Here $a_{i\ell}$ and $b_{j\ell}$ are the group labels in the two checks.
For cyclic labels this condition is a linear congruence on their exponents.
The doubled support $B_X=B_Z=(C\ C)$ is automatically self-orthogonal.
If $C\in\F^{r\times M}$ has full row rank, Theorem~\ref{thm:free} fixes the parameters of every commuting lift over a $2$-group:
\begin{align}
n=2MP,\qquad k=2(M-r)P. \label{eq:parameters}
\end{align}
The full-rank supports guarantee full regularity, while their dimensions fix the rate $1-r/M$; the group labels remain available for improving distance subject to the pair-cancellation constraints.

A reflection lets the same seed design serve both CSS sectors.
For an Abelian $2$-group, take
\begin{align}
H_X&=(U\ V),\qquad H_Z=(\overline V\ \overline U), \label{eq:reverse}\\
H_XH_Z^\dagger&=UV^{\mathsf T}+VU^{\mathsf T}=0. \label{eq:reverse-commutation}
\end{align}
The qubit permutation $\varphi:(\ell,g)\mapsto(\pi(\ell),g^{-1})$, where $\pi$ exchanges the two column halves, sends one check space to the other and gives $d_X=d_Z$.
This fixed-point-free, check-exchanging reflection is the reversing symmetry used for logical pairing and symplectic halving in Ref.~\cite{LeeEtAl2026PP}.
One can therefore search for seeds in the $Z$ sector and obtain $X$ representatives of equal weight by reflection.
Choosing canonical partners still requires the pairing induced by this reflection to be normalized, as in Sec.~\ref{sec:light-partners}.
The next construction controls both seed weights and normalization by separating the two sectors' correction supports.

\subsection{Cofactor frames and canonical pairing}
\label{sec:quaternary}
Using the same quaternary coefficient field and binary trace-basis expansion as in Sec.~\ref{sec:lp-explicit-frames}, we construct fully populated PP checks with full-rank augmentation.
As in the LP construction, we place the two sectors' corrections on disjoint coordinates, leaving only the free coordinates to determine their pairing [Fig.~\ref{fig:unified-frames}(b)].
Here the free-coordinate entries carry pivot determinants, so canonical pairing requires an additional scalar normalization.
Reversing symmetry makes this scalar a determinant square; requiring that square to be a group translation makes normalization preserve binary weight.

Work over $R=\mathbb F_4[C_P]$, with $P=2^a$.
The involution reverses the group element and fixes field coefficients.
Commutation requires the coefficient-weighted contributions at each group label to sum to zero; the lift constraints below preserve these cancellations.
In each CSS sector, commuting $3\times8$ checks with full-rank augmentation give a free logical module of rank two over $R$.
Binary expansion therefore gives $n=16P$, $k=4P$, and four binary seed pairs.

We place the corrections to the $Z$-type seeds on columns $\mathcal P=(0,1,2)$ and those to the $X$-type seeds on the disjoint columns $\mathcal P'=(4,5,6)$.
The remaining columns $\mathcal J=(3,7)$ carry the shared free coordinates.
Define the pivot blocks $D_X=H_X[:,\mathcal P]$ and $D_Z=H_Z[:,\mathcal P']$, with determinants $\Delta_X=\det D_X$ and $\Delta_Z=\det D_Z$.

Multiplying a free-coordinate entry by the corresponding pivot determinant allows its syndrome to be cancelled with cofactors, without expanding the inverse pivot matrix.
The standard adjugate identity $D\adj(D)=\det(D)I$ gives the following explicit construction.
Collect the seeds into $\widetilde Z=(\widetilde z_3\ \widetilde z_7)$ and $\widetilde X=(\widetilde x_3\ \widetilde x_7)$.
With their rows displayed in the order $(\mathcal P,\mathcal P',\mathcal J)$, these matrices are
\begin{equation}
\begin{pmatrix}\widetilde Z & \widetilde X\end{pmatrix}
=
\begin{pmatrix}
\adj(D_X)H_X[:,\mathcal J] & 0_{3\times2}\\
0_{3\times2} & \adj(D_Z)H_Z[:,\mathcal J]\\
\Delta_X I_2 & \Delta_Z I_2
\end{pmatrix}.
\label{eq:cofactors}
\end{equation}
Reflection exchanges the free-column labels and preserves raw seed weights: $\widetilde x_7=\varphi(\widetilde z_3)$ and $\widetilde x_3=\varphi(\widetilde z_7)$.
These disjoint-pivot representatives and their scalar normalization also appear in Ref.~\cite{OkadaKasai2026LWCB}.

\begin{proposition}[Square pairing]
\label{prop:square}
Suppose the reversing checks commute and the augmentation of $D_X$ is invertible.
Then
\begin{align}
\widetilde x_a^\dagger\widetilde z_b=\delta_{ab}g,\qquad g=\Delta_X^2. \label{eq:square}
\end{align}
The seeds $x_a=\widetilde x_a$ and $z_a=\widetilde z_a g^{-1}$ form a full regular canonical frame over $R$.
For $g=t^m$, normalization is a cyclic shift and preserves binary weight.
\end{proposition}
\begin{proof}
The adjugate identity gives $H_X\widetilde z_a=H_Z\widetilde x_a=0$.
An $X$ seed and a $Z$ seed have disjoint coordinate supports unless their labels agree; their only common coordinate is then the free column with that label.
Their pairing is therefore $\delta_{ab}\overline{\Delta_Z}\Delta_X$.
Reversing symmetry gives $D_Z=\overline{D_X}$, hence $\overline{\Delta_Z}=\Delta_X$.
Invertible augmentation makes $\Delta_X$ a unit in the local ring, so normalization gives identity pairing.
The two seed pairs match the rank of each logical module over $R$; expanding each with both trace-basis elements gives four binary regular seed pairs.
\end{proof}
When the scalar pairing instead has the form $ct^m$ with $c\ne1$, normalization can change binary weight; the reported widths include this step and any subsequent paired basis changes.

\input{main_v7_parts/figures/seed_supports}

Period doubling gives a way to vary the code while keeping normalization a fixed translation.
In a lift from period $h$ to $2h$, each exponent $e$ is replaced by $e+hb$ with $b\in\{0,1\}$.
Changing these lift bits changes the lifted minor by a multiple of $1+t^h$, whose square is zero modulo $t^{2h}-1$.
Thus a parent minor $\Delta_h=t^{m_0}$ guarantees $\Delta_{2h}^2=t^{2m_0}$ even when the lifted minor has several terms.
The remaining commutation constraints are affine equations in the lift bits (Appendix~\ref{sec:affine-pp-lifts}).
Their solutions can therefore be searched for distance while canonical normalization remains a fixed translation.
Each lifted code must still satisfy the prescribed physical check and frame weight limits.

With invertible selected pivots, the construction extends to other Abelian groups, including $C_{16}\times C_2$ (Appendix~\ref{sec:small-noncyclic-quaternary}).
For cyclic periods beyond powers of two, the condition is $\gcd(\Delta_X(t),t^P-1)=1$.
The unit pivots give both checks rank $3P$ over $\mathbb F_4$, and the identity pairing makes the two $R$-seed pairs independent; the logical dimension $2P$ over $\mathbb F_4$ then gives completeness.
At $P=48$, a pivot with $\Delta_X^2=1$ gives a $[[768,192,22]]$ code with $w=10$ and $f_{\mathrm{rep}}=29$ (Appendix~\ref{sec:pp-period48}).

The PP block of Table~\ref{tab:code-comparison} exhibits a tradeoff among distance, check weight, and reported frame width.
Among the reported codes with $n=512$, $k=128$, and $d=19$, increasing $w$ from $9$ to $10$ to $11$ lowers the reported width from $30$ to $25$ to $24$.
These short, rate-$1/4$ examples also exceed the $d\le14$ ceiling for the self-adjoint quaternary LP family of Proposition~\ref{prop:small-f4-lp-bound}, which applies at $n\le600$ and rate at least $1/5$.
The $1024$-qubit distance-$24$ examples similarly have reported frame widths $27$ and $25$ at check weights $10$ and $11$, respectively.

Our $[[1024,256,24]]$ PP code matches the distance, rate, check weight, and reported canonical-basis width of the $[[2048,512,24]]$ construction of Okada and Kasai at half the blocklength: both have rate $1/4$, $w=10$, and $f_{\mathrm{rep}}=27$~\cite{OkadaKasai2026LWCB}.
Reversing symmetry makes the cofactor pairing a determinant square; the square-zero changes available in the final lift therefore leave canonical normalization unchanged.
We exploit this freedom through affine commutation constraints to optimize the final exponent lift for distance, subject to the prescribed check and frame weight bounds.
A separate $w=11$ construction attains the same distance with the proved optimum $f=25$.
At comparable blocklength, the $w=10$ PP example also offers a higher encoding rate and a substantially narrower reported canonical basis than the $[[975,195,\le24]]$ mitten code: rate $1/4$ rather than $1/5$, and width $27$ rather than $102$, with check weight $10$ rather than $9$~\cite{BhardwajEtAl2026Processors}.

\subsection{An optimal full regular basis}
\label{sec:pp-width-separation}
We denote by $\mathcal C^{\mathrm{PP}}_{24,11}$ the period-$64$ reversing code with coefficient array $C$ of Eq.~\eqref{eq:c4} and exponent array $E^{(24)}$ of Eq.~\eqref{eq:e24}, using the weight-eleven check presentation specified there.
Its four seed pairs generate all $256$ canonical pairs, as summarized in Fig.~\ref{fig:seeds}(b).

For $\mathcal C^{\mathrm{PP}}_{24,11}$, minimum-weight logical operators cannot span the full logical space.
Exhaustive search shows that every $Z$ logical of weight at most $24$ acts on an even number of qubits within binary blocks $6$ and $15$, the two shaded rows in Fig.~\ref{fig:seeds}(b), whereas some $Z$ logicals have odd parity there.
Multiplication by stabilizers preserves this parity, and products of even-parity operators remain even, so the weight-$24$ operators cannot generate the missing logical classes.
A complete canonical basis therefore requires an operator of weight at least $25$.
Our full regular cofactor frame attains this bound, proving $d=24<f=25$.

\begin{proposition}[Strict separation of $d$ and $f$]
\label{prop:width-separation}
The code $\mathcal C^{\mathrm{PP}}_{24,11}$, defined by the coefficient array~\eqref{eq:c4} and exponents $E^{(24)}$ in Eq.~\eqref{eq:e24}, satisfies
\begin{align}
d=24,\qquad f_{\mathrm{span}}=f=f_{\rm reg}=25. \label{eq:pp-distance-width-separation}
\end{align}
\end{proposition}
\begin{proof}
Let $\lambda(z)$ count the support of $z$ in the union of binary blocks $6$ and $15$, modulo two, as specified in Appendix~\ref{sec:pp-width-detector}.
It vanishes on $Z$ stabilizers and equals one on a cofactor seed, so it induces a nonzero functional on $\LZ$.
The completed exclusion of zero-syndrome $Z$ operators with $\lambda=1$ through weight $24$ gives $d_\lambda\ge25$.
Equation~\eqref{eq:spanning-dual} therefore implies $f_{\mathrm{span}}\ge25$: the classes represented below weight $25$ lie in the proper hyperplane $\ker\lambda$.
For mixed Pauli operators, apply $\lambda$ to the $Z$ component, whose weight cannot exceed the full Pauli weight; projecting a complete basis must span $\LZ$.
The cofactor frame gives $f_{\rm reg}\le25$, so the hierarchy in Eq.~\eqref{eq:frame-width-hierarchy} yields $f_{\mathrm{span}}=f=f_{\rm reg}=25$, including for unrestricted mixed-Pauli canonical bases.
Exhaustive searches on the full code and the weight-$24$ logical operator in Eq.~\eqref{eq:flagship-distance-witness} establish $d=24$.
Appendix~\ref{sec:proof-data-cost} and Ref.~\cite{Lee2026RegularFramesRepository} document the completed exclusions.
\end{proof}

\section{Discussion}
Full regularity, low-weight spanning, and low-weight canonical pairing impose distinct constraints on code design.
The augmentation and residue tests give sufficient criteria for full regularity over finite $2$-groups.
The two LP examples in Sec.~\ref{sec:lp-pairing-main} have minimum-weight regular bases in both sectors but no complete canonical pairing of that weight under the prescribed action; their unrestricted widths remain unresolved.
For $\mathcal C^{\mathrm{PP}}_{24,11}$, minimum-weight logicals fail to span the full logical space, and the regular frame attains the unrestricted optimum $d=24<f=25$.

Controlling syndrome corrections and canonical normalization brings reported frame width close to distance.
Logical spectroscopy constructs complete canonical bases for general Abelian lifted products, but symmetry projectors and pairing transformations can spread their physical support~\cite{Lee2026Spectroscopy}.
Here sparse $D^{-1}E$ corrections give LP seeds with canonical pairing directly; for selected PP constructions, a translation normalizes the pairing without changing weight.
Quaternary coefficients allow these sparse constructions to retain fully populated pivots.

The gains occur at finite blocklength: the $1088$-qubit LP code attains $d=f=22$, and the $512$-qubit PP examples exhibit the tradeoff between check and frame weights.
Matrix dimensions permit higher encoding rates, with distance and frame width to be optimized at each choice.
At a few hundred to about a thousand qubits, rate $1/4$ gives a useful balance for our PP examples.
At small blocklengths, our examples lie below the reported $kd^2/n$ frontier of Ref.~\cite[Sec.~VI.1]{LeeEtAl2026PP}, but supply explicit full regular canonical bases with controlled weights.
Within the $1024$-qubit PP family, attaining $d=24$, $w=10$, and $f=25$ together remains open.

Reversing symmetry also permits non-CSS constructions by folding paired physical coordinates~\cite[Sec.~IV]{LeeEtAl2026PP}.
For $\mathcal C^{\mathrm{PP}}_{24,11}$, this folding preserves the weights of an inherited complete canonical basis.
The inherited translations need not act as physical qubit permutations.

The existence criteria also cover non-Abelian $2$-groups.
Two-sided non-Abelian LP codes use commuting left and right regular representations~\cite{LinPryadko2024,BhardwajEtAl2026Processors,Hong2026ZSZ}.
Unit constituent pivots give canonical partners organized by opposite actions.
Extending these constructions to full regular canonical frames requires a common physical group action preserving both CSS sectors, which remains open for these examples.
The PP construction additionally uses commutative cofactor formulas and a sector-exchange symmetry.

Frame width specifies the support needed to supply a complete canonical logical basis directly.
The obstructions concern such bases, rather than all possible generating sets for logical operations.
Turning the translated supports into fault-tolerant logical measurements requires distance-preserving interfaces and an assessment of their ancilla and circuit costs; surgery and extractor architectures provide concrete settings for this task~\cite{CrossEtAl2024Surgery,WilliamsonYoder2026Gauging,SwaroopEtAl2026Adapters,CowtanEtAl2026Parallel,ZhengJiangXu2025HighRate,HeEtAl2025Extractors}.
Our finite-size examples establish that substantial distance and encoding rate can coexist with complete canonical bases whose weights remain close to the distance.
These results make the complete canonical basis a design objective alongside distance and check weight.

\noindent\begin{minipage}[b]{\columnwidth}
\vspace*{\baselineskip}
\section*{Data availability}
The companion repository~\cite{Lee2026RegularFramesRepository} supplies the full construction catalogue, physical frames and witnesses, independent checkers, search sources, and completion records.
Appendix~\ref{sec:proof-data-cost} distinguishes finite certificates from exhaustive execution evidence; large traversal trees are omitted.
\end{minipage}

\noindent\begin{minipage}[b]{\columnwidth}
\vspace*{\baselineskip}
\section*{Acknowledgements}

J.Y.L. thanks Hengyun Zhou for discussions and feedback, as well as collaboration on Ref.~\cite{LeeEtAl2026PP}. 
The computations were performed through the Illinois Campus Cluster Program at the University of Illinois, Urbana-Champaign. 
J.Y.L. is supported by the faculty startup grant at the University of Illinois, Urbana-Champaign and the IBM-Illinois Discovery Accelerator Institute.
The author acknowledges the use of AI tools from OpenAI and Anthropic for performing simulations, making figures, and editing the text.
The author has reviewed all content and takes full responsibility for the results.
\end{minipage}

\bibliography{main_v7_parts/references}


\appendix
\raggedbottom
\renewcommand{\thesubsection}{\thesection.\arabic{subsection}}
\makeatletter
\renewcommand{\p@subsection}{}
\renewcommand{\p@subsubsection}[1]{\thesubsection(#1)}
\renewcommand{\p@paragraph}[1]{\thesubsection(\thesubsubsection)(#1)}
\makeatother
\input{main_v7_parts/appendices/distance_verification}
\input{main_v7_parts/appendices/lp_theory}
\input{main_v7_parts/appendices/frame_width_obstructions}

\input{main_v7_parts/appendices/construction_catalogue}

\end{document}

%% file: main_v7_parts/figures/module_spectrum.tex
\begin{figure}[!t]
\centering
\includegraphics[width=\columnwidth]{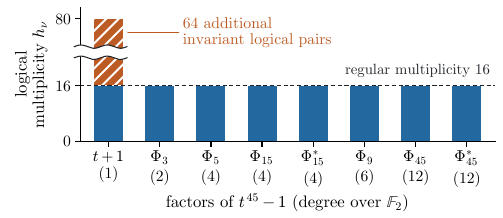}
\vspace{-15pt} 
\caption{ {\bf Obstruction to full regularity}. LP code $\mathsf M^{45}_{5,9}$ from logical spectroscopy~\cite[Sec.~V]{Lee2026Spectroscopy}, with a fully populated monomial $5\times9$ constituent at $P=45$, of rank five outside $t=1$.
The blue part contains $16$ regular copies, accounting for $720$ logical pairs.
The $64$ additional invariant pairs make the component multiplicities unequal, obstructing full regularity.
Labels below each factor give its degree over $\F$.}
\vspace{-10pt}
\label{fig:spectrum}
\end{figure}

%% file: main_v7_parts/figures/unified_frame_construction.tex
\begin{figure*}[t]
\centering
\includegraphics[width=\textwidth]{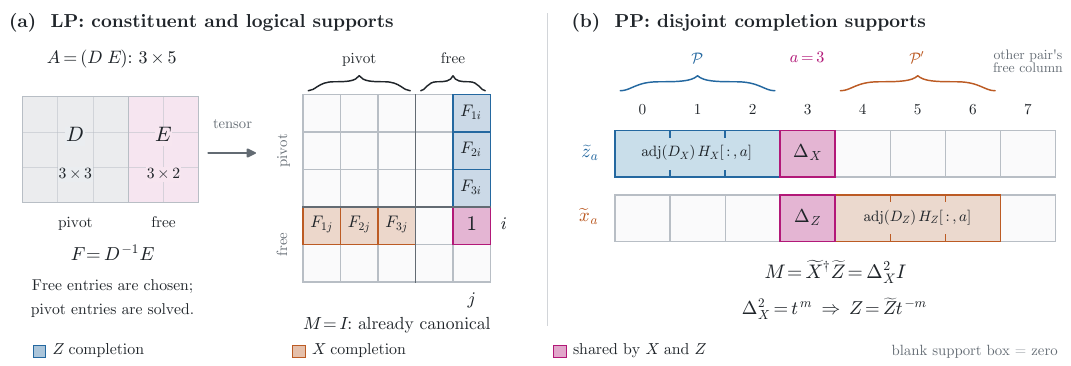}
\caption{ {\bf Canonical pair construction.}
(a) For the $3\times5$ constituent $A=(D\ E)$, only the pivot block $D$ is square; $F=D^{-1}E$ determines the kernel vectors' pivot entries.
The $5\times5$ array represents qubit blocks; the omitted $3\times3$ array is zero, so $H_X$ tests each column and $H_Z$ each row against $A$.
The illustrated seeds share only the unit entry $(i,j)$.
For arbitrary seed labels, $x_{ij}^\dagger z_{kl}=(v_i^\dagger u_k)(u_j^\dagger v_l)=\delta_{ik}\delta_{jl}$, proving $M=I$ for the full frame.
(b) The PP cofactor vectors $\adj(D_X)H_X[:,a]$ and $\adj(D_Z)H_Z[:,a]$ occupy disjoint pivot sets $\mathcal P$ and $\mathcal P'$.
For the illustrated $a=3$ pair, column $7$ is zero; it is the other pair's free column.
Only the common free coordinates contribute to $M=\widetilde X^\dagger\widetilde Z=\Delta_X^2I$ under reversing symmetry.
If $\Delta_X^2=t^m$, normalization $Z=\widetilde Zt^{-m}$ preserves weight, even when $\Delta_X$ has several terms.
Colors distinguish the $Z$ completion, $X$ completion, and shared components.
Each support box contains $|G|$ physical qubits over $\mathbb F_2$, or $2|G|$ after trace expansion over $\mathbb F_4$.}
\label{fig:unified-frames}
\end{figure*}

%% file: main_v7_parts/figures/seed_supports.tex
\begin{figure*}[!t]
\centering
\includegraphics[width=\textwidth]{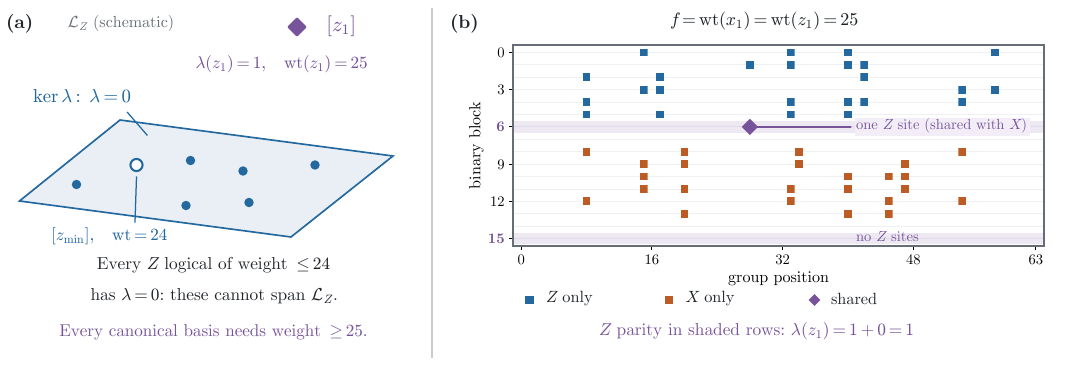}
\caption{ {\bf Separation of $d$ and $f$ for $\mathcal C^{\mathrm{PP}}_{24,11}$ (Proposition~\ref{prop:width-separation}).}
Here $x_1,z_1$ are the binary trace-basis expansions of $\omega x_3,\omega z_3$, where $x_3,z_3$ are the normalized quaternary seeds of Proposition~\ref{prop:square}.
(a) The nonzero logical functional $\lambda$ vanishes on stabilizers and, by exhaustive exclusion, on every $Z$ logical representative of weight at most $24$.
These classes, including the weight-$24$ witness $z_{\min}$, lie in the proper hyperplane $\ker\lambda$ and cannot span $\mathcal L_Z$.
The weight-$25$ seed $z_1$ lies outside it.
Applying $\lambda$ to the $Z$ component also excludes a width-$24$ mixed-Pauli basis.
Logical-space positions are schematic and do not encode weight or physical geometry.
(b) Actual supports of $(x_1,z_1)$: both operators have weight $25$ and share one qubit (purple diamond).
Rows index binary blocks $2\ell+b$, where $b=0,1$ selects the trace-basis component; columns index group positions.
The shaded rows, $6$ and $15$, define $\lambda$; they contain one and zero $Z$ sites, respectively, so $\lambda(z_1)=1$.
All four seed pairs have weight $25$; translating them through $C_{64}$ gives $256$ canonical pairs and proves $f\le25$.}
\label{fig:seeds}
\end{figure*}

%% file: main_v7_parts/appendices/distance_verification.tex
\section{Exhaustive verification and proof data}
\label{sec:distance-algorithm}
To certify a distance $D$, we exhibit a weight-$D$ logical operator and prove that no lighter logical exists.
We check the witness directly and obtain the lower bound by a completed exhaustive search (Fig.~\ref{fig:distance-certificate}).
All searches use binary physical coordinates, including the binary expansion of quaternary checks.

For the $Z$ sector, let $H=\widehat H_X$ and let $S=\im\widehat H_Z^{\mathsf T}$ be the $Z$ stabilizer space.
A vector $z$ commutes with the checks when its syndrome $Hz$ vanishes; it is a nontrivial logical operator when it is not in $S$.
Thus
\begin{align}
d_Z=\min\{\wt(z):Hz=0,\ z\notin S\}.
\end{align}
Exchange $X$ and $Z$ to obtain $d_X$, and take $d=\min(d_X,d_Z)$.
A verified weight-preserving reflection exchanges the sectors in our reversing examples, proving $d_X=d_Z$ and allowing one sector's exclusion to serve for both.
\input{main_v7_parts/figures/distance_search_flow}

\subsection{Building supports and pruning the search}
The search adds occupied coordinates until every check overlaps the support an even number of times~\cite{DumerKovalevPryadko2016}.
It can be restricted to \emph{support-irreducible} kernel words: nonzero words whose supports contain no proper nonempty kernel support.
A shortest logical has this property; otherwise a kernel word on a smaller support would either be a lighter logical or a stabilizer whose removal shortens the original.
Thus these words suffice to exclude $d_Z\le W$.

At each node, $T$ is the nonempty selected support, $s=H\mathbf 1_T$ lists the checks with odd overlap, and the remaining available coordinates are recorded.
If $s\ne0$, choose an unsatisfied check and order its available neighbors as $q_1,\ldots,q_m$.
Every zero-syndrome completion must use at least one neighbor.
Branch $i$ includes $q_i$ and forbids $q_1,\ldots,q_{i-1}$, so every completion belongs to exactly one branch; the ordering affects speed only.
At zero syndrome, test for a logical witness and do not extend that branch, even if the current word is a stabilizer: further additions would contain a proper nonempty kernel support and are unnecessary for finding a shortest logical.

The logical test uses exact binary elimination against the stabilizer row space.
An equivalent test pairs the word with a \emph{complete} opposite-sector logical basis $x_1,\ldots,x_k$:
\begin{align}
\lambda_L(z)&=(x_1^{\mathsf T}z,\ldots,x_k^{\mathsf T}z), \\
z\in\ker H:&\qquad z\in S\ \Longleftrightarrow\ \lambda_L(z)=0. \label{eq:verification-logical-syndrome}
\end{align}
An incomplete basis could miss a logical direction and mistake a logical operator for a stabilizer.

Pruning discards branches that cannot contain a required word.
Let $b=W-|T|$ be the remaining weight budget and $w_c$ the largest column weight of $H$.
Each added coordinate changes at most $w_c$ syndrome bits, so $\wt(s)>bw_c$ rules out completion.
An unsatisfied check with no available neighbor also rules it out.
Any collection of $p$ unsatisfied checks with pairwise disjoint available neighborhoods requires at least $p$ additional coordinates.
Thus $p>b$ permits pruning; the collection need not be maximal.

Stabilizers give a further test.
For every stabilizer generator $g\in S$, a shortest nontrivial logical $z$ must satisfy
\begin{align}
|\operatorname{supp}(z)\cap\operatorname{supp}(g)|\le \frac{\wt(g)}{2}, \label{eq:verification-overlap}
\end{align}
otherwise adding $g$ would reduce its weight without changing its logical class.
As the search adds coordinates, this overlap can only increase.
A partial support that already violates the inequality therefore cannot lead to a shortest logical.
These pruning rules retain a minimum-weight logical whenever $d_Z\le W$; they need not retain every low-weight representative.
For quaternary rows, using both binary components and their sum preserves the check and stabilizer spaces while providing more checks for branching and overlap pruning.

\subsection{Translation symmetry and logical-parity constraints}
Translation symmetry reduces the number of starting coordinates.
For a cyclic lift, write coordinates as $(b,j)$ with block $b$ and position $j\in\mathbb Z/P\mathbb Z$.
Translation preserves weight, the kernel, and the stabilizer space.
Move an occupied coordinate in the first occupied block to position zero.
It suffices to start at $(b,0)$ for every block, forbidding earlier blocks.
This covers every orbit, possibly more than once.

The frame-width certificate asks a different question: how light can a word be if it has a specified nonzero logical parity?
Let the linear functional $\lambda$ vanish on $S$ but be nonzero on the logical quotient, and define
\begin{align}
d_\lambda=\min\{\wt(z):Hz=0,\ \lambda(z)=1\}.
\end{align}
A shortest word with $\lambda=1$ is also irreducible.
A proper nonempty kernel subset would split it into two lighter disjoint kernel words, exactly one of which has parity one.
Adding a stabilizer preserves this parity, so the overlap bound~\eqref{eq:verification-overlap} also holds for a minimum $\lambda$-odd word.
The same branching and pruning rules apply, but a zero-syndrome word qualifies only if $\lambda=1$.

For the functional in Eq.~\eqref{eq:detector}, $\lambda(z)$ counts the occupied coordinates in the union of blocks $6$ and $15$, modulo two.
It is therefore invariant under cyclic translation; write $M=\{6,15\}$ for the selected blocks.
A word with $\lambda=1$ must occupy a block in $M$.
Translate a coordinate in its first occupied block \emph{belonging to $M$} to position zero; only roots $(6,0)$ and $(15,0)$ are needed.
At each root, forbid earlier blocks in $M$, but keep earlier blocks outside $M$ available.
Forbidding every earlier block would miss valid words.
These two roots cover every orbit with $\lambda=1$.

\subsection{Candidate selection and completed bounds}
\label{sec:search-strategy}
A lower bound requires every root and branch to finish without a qualifying word.
Parallel workers must cover every subtree in the specified partition; unfinished runs give no exclusion for unvisited branches.
Witness searches alone give only upper bounds.

For LP candidates, we first vary constituent exponents and coefficients under the pivot and check-weight constraints and screen the cofactor and systematic-seed weights.
During screening, we optimize the constituent-kernel distance, using $sP$ binary coordinates or $2sP$ after quaternary expansion.
A search through $D-1$ rejects a candidate as soon as it finds a lighter kernel word.
Completed exclusions and an attaining kernel word determine this kernel distance exactly.
Final certification must still use the full physical code, because constituent and quantum distances can differ (Appendix~\ref{sec:lp-metric}).
The full-code logical test includes every logical direction, including the $16$ invariant directions of the odd-period example in Appendix~\ref{sec:lp-odd-inputs}.
These certificates complement randomized searches~\cite{PryadkoShabashovKozin2022}, algebraic analyses~\cite{RaveendranDeclercqVasic2025}, and general exact solvers~\cite{WebsterJacobHiggott2026Distance}.

\subsection{What the proof data establish}
\label{sec:proof-data-cost}
The repository~\cite{Lee2026RegularFramesRepository} distinguishes explicit witnesses, finite certificates, and exhaustive-search records.
For Proposition~\ref{prop:lp-regular-pairing-cost}, it supplies the $45$ and $75$ normalized physical orbit masks, logical coordinates, pairing polynomials, graphs, and colorings.
An independent checker rebuilds the graphs and all relative-shift support unions, proving the obstruction \emph{if the pools are complete}.
The enumeration sources and completed root records establish that premise; fresh runs reproduce the pools and node counts.
The $\mathcal C^{\mathrm{LP}}_{16,9}$ enumeration disables stabilizer-overlap pruning to retain all required representatives; Appendix~\ref{sec:lp-pairing-obstruction} justifies completeness through its cutoff.

For $\mathcal C^{\mathrm{PP}}_{24,11}$, the two-root partition covers the completed $\lambda=1$ exclusion through weight $24$.
A partition checker verifies that each subtree is assigned exactly once; separate algebraic checks verify the functional and its weight-$25$ witness.
The distance-$24$ claim uses a separate full-code exclusion and the witness of Eq.~\eqref{eq:flagship-distance-witness}.
The partition check and algebraic witness checks do not verify the exhaustive traversal; the exclusions rely on completed executions of the supplied engine.
Inputs, source hashes, commands, partitions, and completion records permit reproduction, while large traversal trees are omitted.
These are execution records, not formal proof traces independent of the search implementation.

The full-code exclusion through weight $21$ for the width-$22$ $[[1088,128,22]]$ LP code visited $7.82\times10^8$ nodes and completed in $258$\,s wall time with four workers.
Fresh single-process enumerations for $\mathcal C^{\mathrm{LP}}_{16,9}$ and $\mathcal C^{\mathrm{LP}}_{20,10}$ reported elapsed times of $36.4$\,s and $7.4$\,min, respectively, on an Apple M3 Pro.
For the PP $\lambda=1$ exclusion through weight $24$, the runtimes of the completed single-threaded search tasks sum to about $26.7$ hours.
Exact node counts, retained build settings, and completion records are supplied in Ref.~\cite{Lee2026RegularFramesRepository}; processor models and compiler versions were not retained for the distance-$22$ LP run or the PP run.

%% file: main_v7_parts/figures/distance_search_flow.tex
\begin{figure}[!ht]
\centering
\begin{tikzpicture}[
  box/.style={draw=black!55,rounded corners=2pt,align=center,
              font=\footnotesize,inner sep=4pt},
  arr/.style={-{Stealth[length=1.7mm]},line width=.65pt}
]
\node[box,text width=3.25cm,minimum height=1.5cm,fill=blue!4]
  (upper) at (-1.9,0)
  {Weight-$D$ witness\\zero syndrome\\nonzero logical class\\$d\le D$};
\node[box,text width=3.25cm,minimum height=1.5cm,fill=blue!4]
  (lower) at (1.9,0)
  {Complete exclusion\\through weight $D-1$\\in both CSS sectors\\$d\ge D$};
\coordinate (join) at (0,-1.02);
\node[box,text width=3.4cm,fill=blue!10] (exact) at (0,-1.55)
  {Exact distance: $d=D$};
\draw (upper.south) |- (join);
\draw (lower.south) |- (join);
\draw[arr] (join) -- (exact.north);
\end{tikzpicture}
\caption{The two parts of an exact-distance certificate. The exclusion must finish at every required search root; a verified symmetry can transfer it between CSS sectors. An unfinished search does not establish the lower bound.}
\label{fig:distance-certificate}
\end{figure}
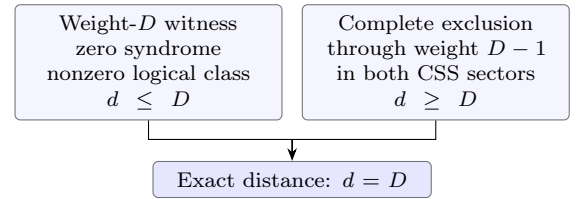

%% file: main_v7_parts/appendices/lp_theory.tex
\section{Algebraic criteria and physical-weight bounds}
\label{sec:lp-theory}
This appendix collects the algebraic criteria and physical-weight bounds used in the constructions.
Constituent ranks determine full regularity; pivot supports and field coefficients bound check and seed weights and limit distance.
The binary expansion conventions apply to both LP and PP codes.

\subsection{Constituent ranks and lift periods}
\label{sec:other-periods}
\label{sec:lp-odd-inputs}
\label{sec:odd-lp-comparison}
For an Abelian group $G$ of odd order, write $\R=\prod_\nu K_\nu$ and consider the LP checks of Eq.~\eqref{eq:lp}.
Set $a_\nu=\rank_{K_\nu}A_\nu$ and $b_\nu=\rank_{K_\nu}B_\nu$.
Invertible changes of basis put each constituent into rank normal form over $K_\nu$ and split off components with zero logical quotient.
This decomposition determines the logical dimensions; the changes of basis need not preserve physical weight.
The remaining kernel--cokernel tensor products give the component multiplicity
\begin{align}
h_\nu&=(n_a-a_\nu)(m_b-b_\nu) \notag\\
&\quad+(m_a-a_\nu)(n_b-b_\nu). \label{eq:lp-components}
\end{align}
Thus $h_\nu$ is constant exactly when a full regular canonical frame exists; logical spectroscopy supplies a complete paired binary basis even when the multiplicities differ~\cite[Sec.~IV.2]{Lee2026Spectroscopy}.

For the self-adjoint binary family $\mathrm{LP}(A,A^\dagger)$, choose $B=A^\dagger$ and $G=C_P$ with odd $P>1$.
Let $A\in\R^{r\times s}$, $s>r\ge2$, contain one monomial in every entry and have rank $r$ at every nontrivial component.
Its augmentation is the all-ones matrix of rank one.
The rank of $(A^\dagger)_\nu$ equals the rank of $A$ at the inversion-related component $\nu^\vee$, so both ranks are $r$ in every nontrivial component and one in the augmentation component.
Equation~\eqref{eq:lp-components} then gives $h_1=(s-1)^2+(r-1)^2$ and $h_\nu=(s-r)^2$ elsewhere, yielding
\begin{align}
\LZ&\cong\R^{(s-r)^2}\oplus\F^{\delta},\qquad k=P(s-r)^2+\delta, \\
\delta&=2(r-1)(s-1)>0. \label{eq:lp-defect}
\end{align}
The free summand supplies regular orbits for $P(s-r)^2$ logical pairs; the remaining $\delta$ pairs are invariant under translation~\cite[Sec.~IV~A]{ZhengEtAl2026CanonicalLP}.

For the $\mathsf M^{45}_{5,9}$ example in Fig.~\ref{fig:spectrum}, $r=5$, $s=9$, and $P=45$ give
$\LZ\cong\mathbb F_2[C_{45}]^{16}\oplus\mathbb F_2^{64}$ and $k=784$~\cite[Sec.~V]{Lee2026Spectroscopy}.

For the published fully populated $3\times5$ constituent at $P=33$~\cite[Eq.~(S74)]{ZhouEtAl2026AlgebraicSurgery}, these ranks give
$\LZ\cong\mathbb F_2[C_{33}]^4\oplus\mathbb F_2^{16}$.
The resulting $[[1122,148,20]]$ code has $132$ logical pairs in regular orbits and $16$ additional invariant pairs, preventing a full regular frame for this action. Our exhaustive certificate establishes the exact distance $d=20$.
The repository provides the constituent, logical witness, and completed search record~\cite{Lee2026RegularFramesRepository}.

A sufficient full-regularity criterion also applies to cyclic groups whose orders are not powers of two.
For $P=2^a m$ with $a\ge1$ and $m$ odd, factor $t^m-1=\prod_\nu p_\nu$ into irreducible polynomials over $\F$.
Since $t^P-1=(t^m-1)^{2^a}$, the group algebra decomposes as $\R\cong\prod_\nu\F[t]/(p_\nu^{2^a})$.
Full row rank of both checks modulo every $p_\nu$ gives $\LX\cong\LZ\cong\R^{L-r_X-r_Z}$: the involution permutes the local factors, and the splitting argument of Theorem~\ref{thm:free} gives the same free rank on each.
Augmentation tests only $p_\nu=t+1$.
For $P=6$, the element $1+t+t^2$ evaluates to $1$ at $t=1$ but divides $t^6-1$ and is not invertible.
When $P$ is a power of two, $t+1$ is the only irreducible factor, so the augmentation test suffices.
The same argument applies over $\mathbb F_4$, with factorization over that field.

\subsection{Sparse pivots and the binary distance ceiling}
\label{sec:lp-design-details}
Consider a three-row binary monomial-or-zero constituent over an Abelian $2$-group, with full augmentation rank and $s>3$.
Choose three columns whose augmented submatrix is invertible and adjoin any nonpivot column.
Form a four-component vector whose $j$th entry is the determinant of the $3\times3$ matrix obtained by deleting column $j$~\cite{SmarandacheVontobel2012}.
The resulting vector $u$ lies in the kernel and is nonzero, since its nonpivot entry is the unit pivot determinant.
Its tensor product with any column of $V$ gives a nontrivial logical operator: writing $u=Uc$, its canonical pairings recover the nonzero coefficient vector $c$.
The binary support matrix of this $3\times4$ submatrix has rank three, so its zero-pattern matrix, obtained by adding the all-ones matrix, has rank at least two.
Consequently, this submatrix has two zero positions in distinct rows and columns.
Among the $4!$ potential permutation terms across all four minors, each zero excludes $3!$ terms, and $2!$ terms involve both zeros.
Inclusion--exclusion therefore leaves at most $4!-2\cdot3!+2!=14$ terms, proving $\wt(u)\le14$ and $d\le14$.

For cyclic lifts with a monomial pivot determinant, $D^{-1}=(\det D)^{-1}\adj(D)$ translates the cofactor supports.
For a binary $3\times3$ pivot whose only zeros are at $(0,0)$ and $(1,1)$ and whose other entries are monomials $t^{e_{ij}}$, impose
\begin{align}
e_{01}+e_{12}+e_{20}=e_{02}+e_{10}+e_{21}\pmod P. \label{eq:lp-pivot-cancel}
\end{align}
Two determinant terms cancel, leaving one monomial.
For a fully populated monomial column of $E$, the three entries of $D^{-1}E$ contain at most $4$, $4$, and $3$ terms, giving a seed-weight bound $1+4+4+3=12$.
These support counts are upper bounds: terms with the same group label can cancel and reduce the physical weight.

\subsection{Quaternary expansion and physical weights}
\label{sec:binary-expansion}
Use the self-dual trace basis $(\omega,\omega^2)$ of $\mathbb F_4=\F(\omega)$, with $\omega^2+\omega+1=0$.
Multiplication by field elements is represented by
\begin{align}
\rho(1)=I_2,\quad \rho(\omega)=\begin{pmatrix}0&1\\1&1\end{pmatrix},\quad \rho(\omega^2)=\begin{pmatrix}1&1\\1&0\end{pmatrix}. \label{eq:rho}
\end{align}
These matrices are symmetric and respect multiplication.
Moreover, $\Tr(\beta_i\beta_j)=\delta_{ij}$ for $(\beta_0,\beta_1)=(\omega,\omega^2)$, so the binary dot product equals the trace of the field pairing.
Binary expansion therefore preserves commutation, and each quaternary canonical seed pair gives two binary canonical pairs by multiplying both partners by $\beta_0$ or $\beta_1$ before taking coordinates.
All weights count physical binary coordinates, not nonzero field symbols.

For a quaternary vector, first combine terms at equal group labels in each coordinate block.
Let $N_1,N_\omega,N_{\omega^2}$ count its remaining nonzero coefficients and set $M=N_1+N_\omega+N_{\omega^2}$.
Its three nonzero scalar multiples have binary weights
\begin{align}
M+N_1,\qquad M+N_{\omega^2},\qquad M+N_\omega,
\label{eq:f4-physical-weights}
\end{align}
for multiplication by $1$, $\omega$, and $\omega^2$, respectively.
For a systematic LP kernel column, the two selected binary seeds consequently have maximum weight $M+\max(N_\omega,N_{\omega^2})$.
The lightest scalar multiple is a nontrivial logical witness of weight $M+\min(N_1,N_\omega,N_{\omega^2})$.

The same counts control check weights.
The two binary rows of a quaternary check and their sum have weights $M+N_{\omega^2}$, $M+N_\omega$, and $M+N_1$.
Any two span the same binary check space, so choosing the two lightest minimizes their maximum weight to
\begin{align}
M+\operatorname{median}(N_1,N_\omega,N_{\omega^2}).
\label{eq:f4-check-balancing}
\end{align}
This pairwise row choice preserves the code and all physical logical operators; column weights must be recomputed for the chosen check basis.
It does not optimize over arbitrary stabilizer row operations.

\subsection{Quaternary pivots and check-weight constraints}
\label{sec:lp-quaternary-design}
\label{sec:lp-check-floor}
Over $\mathbb F_4[C_P]$, a fully populated pivot can have a monomial determinant through coefficient cancellations.
For example, for any integer $m$ set
\begin{align}
D=\begin{pmatrix}
1&1&1\\1&t^m&\omega t^{2m}\\1&\omega^2t^{2m}&t^m
\end{pmatrix}.
\end{align}
The three even-permutation terms have common exponent $2m$ and coefficients $1,\omega,\omega^2$, so they cancel.
Two transposition terms coincide at exponent $m$ and cancel, leaving $\det D=t^{4m}$.
Since $t^{4m}$ is a unit, $D$ is invertible for every cyclic period $P$.
For monomial nonpivot columns, each entry of $\adj(D)E$ contains at most six terms.
After normalization and combination of terms with equal group labels, a systematic kernel column has $M\le19$ nonzero coefficient positions, including its free entry.
The lightest of the three nonzero scalar multiples has binary weight at most their average, $4M/3$, giving the distance ceiling
\begin{align}
d\le M+\min(N_1,N_\omega,N_{\omega^2})\le\left\lfloor\frac{4M}{3}\right\rfloor\le25.
\end{align}
The width $f_{\mathrm{rep}}$ takes the maximum over all selected binary seed weights and can exceed $25$.

For a fully populated constituent, an LP check indexed by $(i,b)$ combines row $i$ and column $b$ of $A$ in disjoint tensor blocks.
Because the involution fixes field coefficients, Eq.~\eqref{eq:f4-check-balancing} applies to this multiset.

\begin{proposition}
For a fully populated monomial $3\times s$ constituent with $s\ge4$ and augmentation rank at least two, the maximum binary check weight is at least $s+5$, both for the raw trace-basis expansion and after the pairwise balancing above.
The bound is attained by full-rank constituents with monomial pivot determinant.
\end{proposition}
\begin{proof}
Every check has $M=s+3$ nonzero coefficients.
The augmentation rank assumption forces a nonconstant coefficient row.
Since there are three nonzero field elements and $s\ge4$, some coefficient $a$ occurs at least twice in that row.
Choose a position $(i,b)$ with coefficient $c\ne a$.
The corresponding check contains both copies of $a$ from the row and two copies of $c$, one from each tensor block.
Its median coefficient count is therefore at least two, proving the balanced bound and hence the raw bound.
For attainment, take the pivot in Appendix~\ref{sec:lp-quaternary-design} and give every nonpivot entry coefficient $1$.
Only the pivot entries with coefficients $\omega$ and $\omega^2$ differ from $1$; every check contains at most two of either, so the raw maximum is $s+5$.
The pivot determinant is $t^{4m}$, giving full augmentation rank.
\end{proof}

The smallest achievable value of $w$ in these binary check presentations is therefore $10$ for $s=5$ and $11$ for $s=6$.
Zeros can reduce check weights; zeros in a nonpivot column also lower the cofactor upper bound on distance.
If a nonpivot column has $q$ nonzero monomial entries, its systematic kernel column has at most $1+6q$ nonzero coefficient positions, because each of the three corrected coordinates is a sum of $q$ cofactors, each with at most two terms.
Scalar averaging therefore bounds the distance by $\lfloor4(1+6q)/3\rfloor$: the cases $q=3,2,1$ give $25,17,9$, respectively.
These bounds use the chosen monomial-determinant pivot; collisions can reduce the actual coefficient counts further.

\input{main_v7_parts/appendices/small_lp_bound}

\subsection{Constituent distance is an upper bound}
\label{sec:lp-metric}
Let $A=(D\ E)$ be an $r\times s$ matrix over a commutative group algebra with $D$ invertible, and use $U,V$ from Eq.~\eqref{eq:lp-regular-seeds}.
For any nonzero $u=Uc\in\ker A$, the cycle $(u\otimes v_j,0)$ has weight $\wt(u)$ and pairs with the systematic $X$ seeds to recover the nonzero vector $c$.
Consequently
\begin{align}
d\bigl(\mathrm{LP}(A,A^\dagger)\bigr)\le d(\ker A).
\end{align}
This statement also holds for quaternary coefficients with binary trace-coordinate weights.
Writing a zero-syndrome vector as an $s\times s$ matrix $Z_1$ and an $r\times r$ matrix $Z_2$ gives $AZ_1=Z_2\overline A$.
Stabilizer vectors have the form $(Y\overline A,AY)$ for an $s\times r$ matrix $Y$.
If $S$ is a right inverse of $A$, adding the stabilizer with $Y=SZ_2$ removes the second block.
This changes the physical representative and can increase its weight.

For a concrete counterexample, take $A=(1\ h)$ over $\F[t]/(t^4-1)$, where $h=1+t+t^3$.
Here $\overline h=h$ and $h^2=1$.
The kernel consists of $(hc,c)$ and has distance $4$: a monomial $c$ attains $4$; weight-two $c$ has nonzero even-weight $hc$; and weight at least three gives total weight at least four.
The quantum code nevertheless has the weight-three logical operator $(Z_1,Z_2)=(I_2,I_1)$, with unit pairing against the systematic $X$ seed.
The expanded check matrix has nonzero, pairwise distinct columns, excluding zero-syndrome words of weights one and two.
Sector exchange therefore gives an exact $[[20,4,3]]$ code.
Choosing $h$ as the pivot also gives a nonmonomial pivot with a sparse inverse, since $h^{-1}=h$.

%% file: main_v7_parts/appendices/small_lp_bound.tex
\subsection{A distance ceiling for short quaternary LP codes}
\label{sec:small-lp-bound}
\begin{proposition}
\label{prop:small-f4-lp-bound}
Let $A\in\mathbb F_4[C_P]^{r\times s}$ have full-row-rank augmentation, $1\le r<s$, and $P=2^a$ with $a\ge0$.
Suppose each entry is zero or a scalar monomial $ct^e$.
If its binary self-adjoint LP code has $n\le600$ and $k/n\ge1/5$, then $d\le14$.
\end{proposition}
\begin{proof}
The weight-preserving logical embedding in Appendix~\ref{sec:lp-metric} bounds the quantum distance by that of the binary constituent code
\begin{align}
\mathcal C=\ker A,\qquad [N,K]=[2sP,2(s-r)P].
\label{eq:small-f4-constituent}
\end{align}
Choose a pivot $D$ with invertible augmentation and adjoin one free column.
The resulting $r+1$ cofactors give a nonzero kernel vector $u$, since its free entry is the unit $\det D$, with at most $(r+1)!$ nonzero quaternary coefficients.
Each nonzero coefficient has binary weights $1,1,2$ under multiplication by the three elements of $\mathbb F_4^\times$.
Some scalar multiple of $u$ therefore has binary weight at most $\lfloor4\operatorname{wt}_4(u)/3\rfloor$, giving $d\le2$ for $r=1$ and $d\le8$ for $r=2$.

For $r\ge3$, Eq.~\eqref{eq:lp-f4-parameters} makes the rate condition equivalent to $s\ge2r$, since
$5(s-r)^2-(s^2+r^2)=2(2s-r)(s-2r)$.
Thus $n\ge10Pr^2$ and $Pr^2\le60$.
Because $P$ is a power of two, $P\in\{1,2,4\}$ and $rP\le12$.
The constituent code has redundancy $N-K=2rP\le24$ and dimension $K\ge2rP\ge6$.
Distance at least $15$ would contradict the binary Griesmer bound~\cite{Griesmer1960}, which requires
\begin{align}
N\ge\sum_{i=0}^{K-1}\left\lceil\frac{15}{2^i}\right\rceil
=15+8+4+2+(K-4)=K+25.
\end{align}
Hence its distance, and therefore the quantum distance, is at most $14$.
\end{proof}

%% file: main_v7_parts/appendices/frame_width_obstructions.tex
\section{Frame-width obstructions}
\label{sec:frame-width-obstructions}
A logical functional can obstruct low-weight spanning, whereas incompatible orbit pairings can obstruct a low-weight regular canonical frame even when each sector has a minimum-weight regular basis.
We give the physical tests supporting these two mechanisms.

\subsection{The functional giving the flagship frame-width bound}
\label{sec:pp-width-detector}
For $\mathcal C^{\mathrm{PP}}_{24,11}$, let $\lambda(z)$ count the occupied coordinates in the union of binary component blocks $6$ and $15$, modulo two, where quaternary column $\ell$ expands into blocks $2\ell,2\ell+1$.
Denote the quaternary coordinate of $z$ in column $\ell$ by $z^{(\ell)}$; then
\begin{align}
\lambda(z)=\Tr\!\left(\omega\aug(z^{(3)})+\omega^2\aug(z^{(7)})\right). \label{eq:detector}
\end{align}
The coefficient array in Eq.~\eqref{eq:c4} satisfies $\omega C_{:,7}+\omega^2 C_{:,3}=0$.
Reversing symmetry then implies that every $Z$ stabilizer has even total support in these two blocks, so $\lambda$ vanishes on $\im H_Z^\dagger$ and is unchanged by multiplication by a $Z$ stabilizer.
The first expanded cofactor seed $\omega z_{a=3}$ has column-$3$ augmentation $\omega$ and zero column-$7$ entry, giving $\lambda=\Tr(\omega^2)=1$.
Thus $\lambda$ induces a nonzero functional on the $Z$ logical quotient.
An exhaustive search excludes zero-syndrome $Z$ operators with $\lambda=1$ through weight $24$.
Equation~\eqref{eq:spanning-dual} therefore gives $f_{\mathrm{span}}^Z\ge25$, establishing the lower bound in Proposition~\ref{prop:width-separation}.
The functional's algebraic properties are finite checks; the exclusion depends on the completed exhaustive run described in Appendix~\ref{sec:distance-algorithm} and preserved in Ref.~\cite{Lee2026RegularFramesRepository}.

\subsection{A parity-functional lower bound}
\label{sec:parity-functional}
A simpler functional counts the total support size modulo two.
If every $Z$ stabilizer has even weight, this parity is unchanged by multiplication by a $Z$ stabilizer.
If some $Z$ logical has odd weight, this parity defines a nonzero functional on $\LZ$.
When the code distance $d$ is even, every logical representative of odd weight has weight at least $d+1$.
The dual characterization and width hierarchy, Eqs.~\eqref{eq:spanning-dual} and~\eqref{eq:frame-width-hierarchy}, then give
\begin{align}
f_{\mathrm{span}}^Z\ge d+1,\qquad f\ge d+1.
\label{eq:parity-functional-bound}
\end{align}
This bound includes canonical bases with mixed Pauli representatives: projection onto the $Z$ sector cannot increase weight and must span that sector.
It applies to the $[[1024,256,24]]$ quaternary PP code with check weight $w=10$, whose supplied odd-weight $Z$ seeds certify that the functional is nonzero.

\input{main_v7_parts/appendices/lp_pairing_obstruction}

%% file: main_v7_parts/appendices/lp_pairing_obstruction.tex
\subsection{Canonical pairing under the prescribed cyclic action}
\label{sec:lp-pairing-obstruction}
Independent low-weight bases in the two logical sectors need not admit a low-weight canonical pairing that respects the prescribed group action.
The following examples separate these requirements within the same LP family.
A \emph{full regular symplectic frame} consists of $q=k/P$ Pauli seed pairs whose translates form a complete canonical basis of logical classes modulo stabilizers and phase, with matching translates paired.
The seeds may mix $X$ and $Z$; $f_{\mathrm{reg}}$ remains restricted to CSS frames.

\begin{proposition}\label{prop:lp-regular-pairing-cost}
For each of the $[[1088,128,d]]$ codes $\mathcal C^{\mathrm{LP}}_{16,9}$ and $\mathcal C^{\mathrm{LP}}_{20,10}$, specified in Eqs.~\eqref{eq:lp-pairing-inputs} and~\eqref{eq:lp-pairing-input20}, respectively, each logical sector has a complete regular basis of maximum weight equal to the distance, so $f_{\mathrm{span}}^X=f_{\mathrm{span}}^Z=d$.
For the specified $C_{16}$ action, however,
\begin{align}
f_{\mathrm{reg}}(\mathcal C^{\mathrm{LP}}_{16,9})=18,
\qquad 21\le f_{\mathrm{reg}}(\mathcal C^{\mathrm{LP}}_{20,10})\le22.
\end{align}
The lower-bound exclusions also hold for full regular symplectic frames whose seed operators need not be pure $X$ or pure $Z$.
The unrestricted canonical width remains undetermined; for $\mathcal C^{\mathrm{LP}}_{16,9}$, it is either $16$ or $18$.
\end{proposition}
\begin{proof}
For each input, Ref.~\cite{Lee2026RegularFramesRepository} supplies eight physical $Z$ seeds of weight $d$ whose $16$ translates have zero syndrome and quotient rank $128$ modulo the $Z$ stabilizers.
Tensor-coordinate transposition gives the corresponding $X$ basis.
Direct syndrome and rank checks establish these separate regular bases of the $X$ and $Z$ logical sectors; the distance lower bounds prove their weight optimality.

For the lower bounds on canonical pairing, exhaustively enumerate pure nontrivial $Z$ operators through weight $17$ for $\mathcal C^{\mathrm{LP}}_{16,9}$ and weight $20$ for $\mathcal C^{\mathrm{LP}}_{20,10}$ using Appendix~\ref{sec:distance-algorithm}.
For $\mathcal C^{\mathrm{LP}}_{16,9}$, omit stabilizer-overlap pruning so that all representatives are retained.
Every logical word through weight $17$ is irreducible: a disjoint kernel decomposition would contain a nontrivial component of weight at least $16$ and another nonzero component of weight at least two, since $H_X$ has no zero column.
For $\mathcal C^{\mathrm{LP}}_{20,10}$, every minimum-weight logical word is irreducible, and stabilizer-overlap pruning cannot discard one.
Represent each translation orbit by the support whose binary mask has the smallest integer value under the fixed qubit ordering.
The complete lists of physical orbit representatives, their logical coordinates, and the enumeration commands and completion records are supplied in Ref.~\cite{Lee2026RegularFramesRepository}.
The resulting data are
\begin{center}
\begin{tabular}{@{}>{\raggedright\arraybackslash}m{0.64\columnwidth}rr@{}}
\toprule
 & $\mathcal C^{\mathrm{LP}}_{16,9}$ & $\mathcal C^{\mathrm{LP}}_{20,10}$\\
\midrule
Enumeration weight cutoff & $17$ & $20$\\
Pure $Z$ orbits & $45$ & $75$\\
Weight of every representative & $16$ & $20$\\
$\mu_{XZ}$ of Eq.~\eqref{eq:lp-support-union} & $26$ & $32$\\
Candidate canonical orbit pairs & $177$ & $88$\\
Colors in a verified proper coloring & $5$ & $3$\\
\bottomrule
\end{tabular}
\end{center}
Sector exchange supplies the complete $X$ pools.
For the chosen physical orbit representatives $x_i,z_j$, the support-union minimum includes every relative cyclic translation:
\begin{align}
\mu_{XZ}=\min_{i,j}\min_{0\le h<16}
\bigl|\operatorname{supp}(x_i)\cup\operatorname{supp}(t^h z_j)\bigr|.
\label{eq:lp-support-union}
\end{align}
A common translation preserves this cardinality, so the relative shift covers every pair of physical representatives in the two pools.
The values $26$ and $32$ imply that a mixed logical Pauli operator at or below the respective weight cutoff has at most one CSS component representing a nontrivial logical class.
Removing the other component by multiplication by a stabilizer preserves the logical class and all pairings and cannot increase physical weight.
For a regular frame, choose a pure-$X$ or pure-$Z$ representative for each seed in this way, then translate it throughout its orbit.
The physical group action preserves the CSS sectors and stabilizer spaces, so this operation preserves the orbit organization.
Each canonical orbit pair must then have one pure $X$ and one pure $Z$ partner; exchanging their names if necessary gives a CSS regular frame within the same weight cutoff.
Applied to each operator of any canonical basis, this replacement shows that a width-$17$ unrestricted frame for $\mathcal C^{\mathrm{LP}}_{16,9}$ would give a width-$16$ frame, since every remaining pure representative has weight $16$.
The known width-$18$ frame therefore leaves $f(\mathcal C^{\mathrm{LP}}_{16,9})\in\{16,18\}$.

For orbit representatives $x_i,z_j$, form
\begin{align}
G_{ij}(t)=\sum_{h=0}^{15}B(x_i,t^h z_j) t^h.
\end{align}
A regular canonical pair requires $G_{ij}$ to be a monomial, whose exponent can be removed by translating one partner.
Shifting $x_i,z_j$ to $t^a x_i,t^b z_j$ multiplies $G_{ij}(t)$ by $t^{a-b}$, preserving its zero or monomial character; orbit origins therefore do not affect compatibility.
Make each such pair a graph vertex; join $(i,j)$ and $(i',j')$ precisely when $G_{ij'}=G_{i'j}=0$.
A full regular frame would require a clique of $k/P=8$ vertices.
The verified proper colorings in the table exclude such a clique.
The supplied five- and three-color assignments are checked against every graph edge; optimal colorings are not needed.
The independent checker in Ref.~\cite{Lee2026RegularFramesRepository} reconstructs the graph and relative-shift unions from the physical masks, also checking syndromes, nontriviality, logical coordinates, and the choice of orbit representatives.
These finite checks validate the obstruction conditional on pool completeness; the separate exhaustive runs establish that premise.
The existing canonical frames of widths $18$ and $22$ give the upper bounds.
\end{proof}

Thus the LP lower bounds arise from the compatibility of canonical pairs, even though minimum-weight regular bases already span both sectors.
This differs from the logical hyperplane obstruction for $\mathcal C^{\mathrm{PP}}_{24,11}$ in Proposition~\ref{prop:width-separation}.

%% file: main_v7_parts/appendices/construction_catalogue.tex
\section{Representative constructions and catalogue}
\label{sec:construction-catalogue}
The principal LP and PP inputs below specify the examples used for the optimal-width and pairing-obstruction results.
The versioned catalogue supplies the remaining coefficient and exponent arrays, pivots, physical frames, witnesses, and completed-exclusion records~\cite{Lee2026RegularFramesRepository}.
Every exact distance reported here combines a nontrivial weight-$d$ witness with a completed full-code exclusion through weight $d-1$, in both CSS sectors or with a verified sector-exchange symmetry, as described in Appendix~\ref{sec:distance-algorithm}.
All quaternary constructions use the trace basis and physical-weight conventions of Appendix~\ref{sec:binary-expansion}.

\input{main_v7_parts/appendices/lp_constituents}

\input{main_v7_parts/appendices/pp_reproducibility}
\input{main_v7_parts/appendices/small_quaternary}

%% file: main_v7_parts/appendices/lp_constituents.tex
\subsection{Lifted-product constructions}
\label{sec:lp-inputs}
\label{sec:lp-quaternary-inputs}

Unless stated otherwise, the LP examples use the self-adjoint family $B=A^\dagger$ of Sec.~\ref{sec:lp}, with quaternary inputs expanded as in Appendix~\ref{sec:binary-expansion}.
For each quaternary check, order the binary row choices as its first component, second component, and their sum, and select the two lowest-weight rows, breaking ties in this order at every translate.
For self-adjoint inputs, transposing the tensor coordinates separately in the $s^2$ and $r^2$ blocks exchanges the CSS check spaces, proving $d_X=d_Z$.

\subsubsection{A binary example and the pairing-obstruction inputs}
\label{sec:lp-quaternary-capped-inputs}
The binary $[[1088,128,12]]$ example uses $P=32$ and the exponent array
\begin{align}
\mathcal E_{1088,12}=\begin{pmatrix}
-&0&0&0&0\\0&-&3&2&21\\0&3&13&5&25
\end{pmatrix},
\end{align}
where $e$ denotes $t^e$ and a dash denotes zero.
Its first three columns give a monomial-determinant pivot obeying Eq.~\eqref{eq:lp-pivot-cancel}.
Both systematic kernel columns have weight $12$, so the certified distance gives $d=f=f_{\rm reg}=12$.

The pairing-obstruction codes $\mathcal C^{\mathrm{LP}}_{20,10}$ and $\mathcal C^{\mathrm{LP}}_{16,9}$ of Proposition~\ref{prop:lp-regular-pairing-cost} have the following coefficient and exponent arrays, indexed by $(d,w)$.
Both use $P=16$ and $A_{ij}=C_{ij}t^{E_{ij}}$, with
{\small\setlength{\arraycolsep}{2.75pt}
\begin{align}
C_{20,10}&=\begin{pmatrix}
1&1&1&1&1\\1&1&\omega^2&1&1\\1&\omega&1&\omega^2&1
\end{pmatrix},&
E_{20,10}&=\begin{pmatrix}
0&0&0&0&0\\0&7&14&12&13\\0&14&7&0&10
\end{pmatrix},\label{eq:lp-pairing-input20}\\
C_{16,9}&=\begin{pmatrix}
1&1&1&1&1\\1&1&\omega&0&\omega^2\\1&\omega^2&1&\omega&0
\end{pmatrix},&
E_{16,9}&=\begin{pmatrix}
0&0&0&0&0\\0&6&12&9&3\\0&12&6&2&14
\end{pmatrix}.
\label{eq:lp-pairing-inputs}
\end{align}}
The first three columns have determinants $t^{12}$ and $t^8$, respectively.
For $\mathcal C^{\mathrm{LP}}_{20,10}$, the two kernel columns give binary seed-weight pairs $(20,20)$ and $(22,22)$; for $\mathcal C^{\mathrm{LP}}_{16,9}$, they give $(16,16)$ and $(18,16)$.
Their translated canonical frames therefore have widths $22$ and $18$, providing the upper bounds used in Proposition~\ref{prop:lp-regular-pairing-cost}.

\input{main_v7_parts/appendices/lp_extended_search}

\input{main_v7_parts/appendices/lp_complete_table}

%% file: main_v7_parts/appendices/lp_extended_search.tex
\subsubsection{An optimal distance-22 frame}
The $[[1088,128,22]]$ construction with $w=10$ and $f_{\mathrm{rep}}=22$ uses $P=16$ and the arrays
{\small\setlength{\arraycolsep}{3pt}
\begin{align}
C_{\mathrm{opt}}&=\begin{pmatrix}1&1&1&1&1\\1&1&\omega&1&\omega^2\\1&\omega^2&1&\omega&1\end{pmatrix},&
E_{\mathrm{opt}}&=\begin{pmatrix}0&0&0&0&0\\0&6&12&13&1\\0&12&6&9&5\end{pmatrix}.
\label{eq:lp-q16-arrays}
\end{align}}
The first-three-column determinant is $t^8$.
After combining terms, each systematic kernel column has $M=18$ and $N_\omega=N_{\omega^2}=4$; Eq.~\eqref{eq:f4-physical-weights} gives weight $18+4=22$ for both trace-basis multiples and hence every seed in Eq.~\eqref{eq:lp-regular-seeds}.
The completed full-code exclusion through weight $21$, described in Appendix~\ref{sec:proof-data-cost}, gives $d\ge22$.
The canonical frame supplies weight-$22$ nontrivial logical operators and has reported width $22$, hence $d=f=f_{\mathrm{reg}}=22$.

\subsubsection{Other pivots and check presentations}
\label{sec:lp-short-quaternary}
\label{sec:lp-extended-inputs}
\label{sec:lp-eight}
\label{sec:lp-independent}
The smaller examples show that a sparse canonical frame can use a pivot with a nonmonomial determinant or a lift period that is not a power of two.
The $[[400,16,18]]$ code, with $P=8$, has determinant $\omega^2(1+t)+t^2+t^6+t^7$, a nonmonomial unit whose augmentation is $1$.
The systematic construction using this pivot gives a frame of width $18$, equal to the certified distance.
At odd period $P=7$, the quaternary $[[476,56,13]]$ and $[[476,56,12]]$ codes have monomial pivot determinants $t^2$ and $t^3$, respectively, ensuring invertibility.
Their systematic constructions give full regular frames of widths $18$ and $14$.
These examples lie outside the binary monomial obstruction of Proposition~\ref{prop:odd}.

The other short quaternary entries in Table~\ref{tab:lp-complete} likewise use unit pivots, including nonmonomial determinants; their arrays and frame data are supplied in the repository~\cite{Lee2026RegularFramesRepository}.

Two zero entries in the pivot suffice to lower the maximum check weight to eight while retaining quaternary coefficients and a full regular frame.
The $[[544,64,14]]$ and $[[1088,128,16]]$ codes with $w=8$ use the common coefficient array
\begin{align}
C&=\begin{pmatrix}\omega&0&1&1&1\\0&\omega&1&1&1\\1&1&1&1&1\end{pmatrix},
\label{eq:lp-eight-coefficients}
\end{align}
with periods $8$ and $16$, respectively, and exponent arrays
\begin{align}
E^{(8)}&=\begin{pmatrix}0&0&0&0&0\\0&7&1&0&4\\0&0&3&4&1\end{pmatrix},&
E^{(16)}&=\begin{pmatrix}0&0&0&0&0\\0&4&0&12&8\\0&0&4&10&5\end{pmatrix}.
\label{eq:lp-eight-exponents}
\end{align}
The first three columns have augmentation determinant $\omega^2$, so they are unit pivots at these power-of-two periods.
A check combines a constituent row and column: it has either eight coefficients equal to $1$, seven nonzero coefficients with at most one $\omega$, or six with at most two $\omega$ coefficients.
Both trace-basis rows therefore have weight at most eight.
Equation~\eqref{eq:lp-regular-seeds} gives reported frame widths $22$ and $26$, respectively, at check weight $w=8$.

The check-weight floor for the dense self-adjoint family can be avoided by choosing two constituents independently.
Over $R=\mathbb F_4[C_P]$, take $A,C\in R^{r\times s}$, substitute $B=C^\dagger$ in Eq.~\eqref{eq:lp} and order each constituent's invertible pivot columns first.
Let $U_A$ and $U_C$ denote their systematic kernel matrices, with free-coordinate vectors $v_i$.
The seeds $z_{ij}=(u_{A,i}\otimes v_j,0)$ and $x_{ij}=(v_i\otimes u_{C,j},0)$ again have identity pairing, since each systematic kernel matrix restricts to the identity on its free coordinates.
The binary parameters remain $n=2P(s^2+r^2)$ and $k=2P(s-r)^2$, and the reported width is the largest weight among the trace-basis multiples of the two sets of kernel columns.

An independent-constituent $[[544,64,14]]$ example, distinct from the self-adjoint table entries, uses $P=8$, $r=3$, and $s=5$, with the same exponents in $A$ and $C$ and coefficient arrays related by $\omega\leftrightarrow\omega^2$.
Their first three columns have determinants $\omega t^7$ and $\omega^2t^7$, respectively, and hence give unit pivots.
Each expanded check has weight at most nine: its coefficient multiset contains at most one $\omega$ and one $\omega^2$, with all remaining coefficients equal to $1$.
The complete regular canonical frame has width $16$.

%% file: main_v7_parts/appendices/lp_complete_table.tex
\subsubsection{Parameters and frame widths}
\label{sec:lp-complete-table}
Table~\ref{tab:lp-complete} collects the certified self-adjoint LP examples with full regular frames reported here.
Rows are ordered by $n$, then $d$, then $w$; the field and $(r,s;P)$ specify the constituent dimensions and period.
Equation~\eqref{eq:lp-regular-seeds} supplies their frame seeds from the stated pivots.

\par\medskip\noindent
\begin{minipage}{\columnwidth}
\makeatletter\def\@captype{table}\makeatother
\centering\small\setlength{\tabcolsep}{7pt}
\begin{tabular}{@{}ccrrrr@{}}
\toprule
Field & $(r,s;P)$ & $[[n,k,d]]$ & $w$ & $w_c$ & $f_{\mathrm{rep}}$ \\
\midrule
$\mathbb F_2$ & $(2,4;16)$ & $[[320,64,10]]$ & $9$ & $6$ & ${}^{\star}10$ \\
$\mathbb F_4$ & $(3,4;8)$ & $[[400,16,18]]$ & $9$ & $5$ & ${}^{\star}18$ \\
$\mathbb F_4$ & $(3,4;8)$ & $[[400,16,20]]$ & $9$ & $5$ & $24$ \\
$\mathbb F_2$ & $(4,6;8)$ & $[[416,32,12]]$ & $9$ & $5$ & $13$ \\
$\mathbb F_2$ & $(4,6;8)$ & $[[416,32,12]]$ & $10$ & $6$ & ${}^{\star}12$ \\
$\mathbb F_4$ & $(3,5;7)$ & $[[476,56,12]]$ & $10$ & $6$ & $14$ \\
$\mathbb F_4$ & $(3,5;7)$ & $[[476,56,13]]$ & $10$ & $6$ & $18$ \\
$\mathbb F_4$ & $(3,5;7)$ & $[[476,56,14]]$ & $10$ & $6$ & $20$ \\
$\mathbb F_4$ & $(3,5;8)$ & $[[544,64,14]]$ & $8$ & $5$ & $22$ \\
$\mathbb F_4$ & $(3,5;8)$ & $[[544,64,14]]$ & $9$ & $6$ & $16$ \\
$\mathbb F_4$ & $(3,5;8)$ & $[[544,64,15]]$ & $10$ & $6$ & $22$ \\
$\mathbb F_2$ & $(3,5;32)$ & $[[1088,128,12]]$ & $8$ & $5$ & ${}^{\star}12$ \\
$\mathbb F_4$ & $(3,5;16)$ & $[[1088,128,16]]$ & $8$ & $5$ & $26$ \\
$\mathbb F_4$ & $(3,5;16)$ & $[[1088,128,16]]$ & $9$ & $5$ & $18$ \\
$\mathbb F_4$ & $(3,5;16)$ & $[[1088,128,18]]$ & $10$ & $6$ & $20$ \\
$\mathbb F_4$ & $(3,5;16)$ & $[[1088,128,20]]$ & $10$ & $6$ & $22$ \\
$\mathbb F_4$ & $(3,5;16)$ & $[[1088,128,21]]$ & $10$ & $6$ & $25$ \\
$\mathbb F_4$ & $(3,5;16)$ & $[[1088,128,21]]$ & $11$ & $9$ & $24$ \\
$\mathbb F_4$ & $(3,5;16)$ & $[[1088,128,22]]$ & $10$ & $6$ & ${}^{\star}22$ \\
$\mathbb F_4$ & $(3,6;16)$ & $[[1440,288,17]]$ & $13$ & $10$ & $26$ \\
$\mathbb F_2$ & $(3,7;32)$ & $[[1856,512,12]]$ & $10$ & $7$ & ${}^{\star}12$ \\
$\mathbb F_2$ & $(3,7;32)$ & $[[1856,512,14]]$ & $12$ & $8$ & $16$ \\
$\mathbb F_4$ & $(3,7;16)$ & $[[1856,512,15]]$ & $14$ & $10$ & $25$ \\
$\mathbb F_2$ & $(4,7;32)$ & $[[2080,288,13]]$ & $9$ & $5$ & $14$ \\
\bottomrule
\end{tabular}
\caption{Certified self-adjoint LP codes with full regular frames, in binary coordinates.
Each distance is established by exhaustive exclusion and a logical witness of that weight.
A left superscript $\star$ marks a proved optimum $f_{\mathrm{rep}}=f$.
The maximum check and column weights $w$ and $w_c$ are per CSS sector, with independent check rows.
For each quaternary check, we choose an independent pair of binary rows that minimizes their maximum weight.
The repository maps each row to its construction data and verification records~\cite{Lee2026RegularFramesRepository}.}
\label{tab:lp-complete}
\end{minipage}\par\medskip

%% file: main_v7_parts/appendices/pp_reproducibility.tex
\subsection{Pair-partition constructions}
\label{sec:pp-reproducibility}
The cyclic PP examples use the reversing checks and cofactor seeds of Sec.~\ref{sec:quaternary}, with the pivot and free-column sets specified there.
For these quaternary inputs, full-code distance exclusions cover all sixteen binary translation roots, with complete disjoint branch intervals when a root is partitioned; reflection exchanges the CSS sectors.
The noncyclic example below uses the corresponding product-group symmetry reduction.

\subsubsection{The flagship inputs}
\label{sec:pp1024-input}
At $P=64$, set $(H_X)_{i\ell}=C_{i\ell}t^{e_{i\ell}}$ and use Eq.~\eqref{eq:reverse}, with
{\small\setlength{\arraycolsep}{3.2pt}
\begin{align}
C=\begin{pmatrix}1&\omega^2&1&\omega&\omega&\omega&\omega&\omega^2\\\omega^2&\omega&1&\omega^2&\omega^2&\omega^2&1&1\\\omega^2&1&\omega^2&\omega&1&\omega&\omega&\omega^2\end{pmatrix}. \label{eq:c4}
\end{align}
}
The submatrix formed by the first three columns of $C$ has determinant $1$.
This verifies the full-rank augmentation hypothesis of Sec.~\ref{sec:quaternary}.
Take the exponent arrays
{\small\setlength{\arraycolsep}{3.2pt}
\begin{align}
E^{(24)}&=\begin{pmatrix}0&0&0&0&0&4&6&58\\0&25&50&62&29&24&0&24\\0&50&25&16&10&0&54&31\end{pmatrix}, \label{eq:e24} \\
E^{(25)}&=\begin{pmatrix}0&0&0&0&0&4&38&26\\0&25&18&30&29&24&0&24\\0&18&25&16&42&0&22&63\end{pmatrix}. \label{eq:e4}
\end{align}
}
The array $E^{(24)}$, together with $C$, specifies $\mathcal C^{\mathrm{PP}}_{24,11}$; its weight-eleven check presentation is given below.
Direct substitution verifies CSS commutation for both exponent arrays and gives the minors
\begin{align}
\Delta^{(24)}=t^{36},\qquad \Delta^{(25)}=t^{18}+t^{36}+t^{50}. \label{eq:minors}
\end{align}
Their difference is $t^{18}(1+t^{32})$, whose square vanishes, so both squared pairings equal $t^8$ and use the same normalizing shift.
For $E^{(24)}$, each cofactor column has $M=19$ and $N_\omega=N_{\omega^2}=6$ after combining terms; both trace-basis multiples therefore have weight $19+6=25$ by Eq.~\eqref{eq:f4-physical-weights}, and the normalizing translation preserves it.
The cofactor formula gives full regular canonical frames of widths $25$ and $27$, respectively; their translated binary pairing matrices are $I_{256}$.

For $E^{(24)}$ and $E^{(25)}$, the following invertible row operations give binary check presentations with maximum weight $11$.
For each pair of component rows $(a,b)$ in quaternary rows $0,1,2$, replace it by $(a,a+b)$, $(b,a+b)$, and $(b,a+b)$, respectively, at every cyclic position in both sectors.
These operations preserve the code, distance, and logical frame; the maximum column weight is five.

For an explicit distance witness for $\mathcal C^{\mathrm{PP}}_{24,11}$, number the binary coordinates by $q=(2\ell+b)64+j$, where $0\le\ell<8$, $b=0,1$ indexes $(\omega,\omega^2)$, and $0\le j<64$.
A nontrivial $Z$-type logical operator has support
\begin{align}
\operatorname{supp}(z)=\{\,&15,54,60,79,99,100,114,124,\notag\\
&138,153,178,188,202,217,221,246,\notag\\
&271,285,306,330,335,345,355,548\,\}.
\label{eq:flagship-distance-witness}
\end{align}
Direct binary linear algebra gives $\widehat H_Xz=0$ and $z\notin\im\widehat H_Z^{\mathsf T}$.
Its weight $24$ matches the completed lower-bound exclusion documented in Appendix~\ref{sec:distance-algorithm}; the width-obstructing functional is given in Appendix~\ref{sec:pp-width-detector}.

\subsubsection{Affine constraints for the final lift}
\label{sec:affine-pp-lifts}
The arrays $E^{(24)}$ and $E^{(25)}$ reduce to the same period-$32$ parent.
Their commuting period-$64$ lifts are parameterized by affine equations over $\F$, as follows.
Write $e_j=e_j^{(0)}+hb_j$ at period $2h$, with $b_j\in\F$.
Each exponent $\Lambda_\alpha(e)$ in a commutation expression is an integer linear form in the entries of $e$.
For such an expression $F=\sum_\alpha c_\alpha t^{\Lambda_\alpha(e)}$ vanishing at period $h$, group terms by $c=\Lambda_\alpha(e^{(0)})\bmod h$.
At period $2h$, terms with the same residue $c$ separate between exponents $c$ and $c+h$.
Their combined coefficient already vanishes in the parent construction, so the two coefficients are equal and commutation requires either one to vanish.
Set $\kappa_\alpha=\lfloor \Lambda_\alpha(e^{(0)})/h\rfloor\bmod2$ and $m_\alpha\cdot b=\Lambda_\alpha(b)\bmod2$.
Setting the coefficient of $t^{c+h}$ to zero gives
\begin{align}
\sum_{\alpha:\,\Lambda_\alpha(e^{(0)})\equiv c\ (h)}c_\alpha(\kappa_\alpha+m_\alpha\cdot b)=0. \label{eq:affinelift}
\end{align}
Each field equation gives two affine binary equations.
Applying the same calculation to a prescribed monomial minor imposes affine constraints at intermediate levels.
At the final doubling step, a parent minor $t^{m_0}$ already guarantees a squared pairing $t^{2m_0}$, so no further determinant constraint is needed to keep normalization a translation.

%% file: main_v7_parts/appendices/small_quaternary.tex
\subsubsection{Shorter cyclic and noncyclic constructions}
\label{sec:additional-quaternary}
\label{sec:small-quaternary}
\label{sec:small-512-quaternary}
\label{sec:pp-eight}
\label{sec:small-noncyclic-quaternary}
At $P=16$, a pivot determinant $\Delta_X=t^9$ gives squared pairing $g=t^2$, so normalization preserves weight.
The resulting $[[256,64,12]]$ code has a full regular frame of width $19$.

The five fully populated period-$32$ constructions in Table~\ref{tab:quaternary} use $3\times8$ monomial checks over $\mathbb F_4[C_{32}]$.
Four binary seed pairs generate all $128$ canonical pairs under translation.
Their determinant squares have the form $ct^m$.
A nontrivial field scalar can change binary weight; the tabulated widths are evaluated after normalization and the stated paired basis changes.

The fully populated $3\times8$ family with a unit augmentation pivot cannot attain check weight eight by changing the two binary trace rows of each quaternary check.
Every nonzero field-multiplication block contributes at least one entry to each binary row.
Equality at eight would make every transformed block a permutation matrix.
Distinct field coefficients have a ratio of order three, whereas the ratio of two binary $2\times2$ permutation matrices has order one or two.
Thus each coefficient row would be constant and the augmentation would have rank at most one.
This argument concerns the stated trace presentation, rather than arbitrary stabilizer generating sets.

Allowing one zero in each coefficient row avoids that restriction.
The remaining entries, five $1$'s, one $\omega$, and one $\omega^2$, give binary check weight eight.
Explicit period-$32$ and period-$64$ inputs give $[[512,128,11]]$ and $[[1024,256,12]]$, respectively, both with maximum column weight four.
After group-algebra pairing normalization, their full regular frames have widths $16$ and $15$.
These examples retain rate $1/4$ but have smaller distance than the constructions with check weights $w=9$ or $w=10$ in Table~\ref{tab:quaternary}.

A distinct $[[512,128,19]]$ construction uses $G=C_{16}\times C_2=\langle s,u\rangle$, with $s^{16}=u^2=1$ and componentwise addition of group labels $(a,b)$ modulo $16$ and $2$.
For the coefficient and group-label arrays supplied in the repository, the same pivot columns give
\begin{align}
\Delta_X=u(s^6+s^5)+s^{14},\qquad \Delta_X^2=s^{10}.
\end{align}
Normalization by $s^6$ preserves weight.
All $32$ group translates give cross-Gram matrix $I_{128}$ and frame width $25$, with $w=10$ and $w_c=4$, matching the parameters of one cyclic row in Table~\ref{tab:quaternary}.
The exclusion uses cyclic-subgroup period $16$ and the independently verified $u$ symmetry to cover sixteen translation-root representatives of the full group action.
The product-group checker verifies the physical frame and these symmetries~\cite{Lee2026RegularFramesRepository}.

\subsubsection{A cyclic period that is not a power of two}
\label{sec:pp-period48}
At $P=48$, the $[[768,192,20]]$ construction with check weight $w=9$ has $f_{\mathrm{rep}}=30$ and pivot determinant $\Delta_X=1+t^{24}+\omega^2t^{25}$.
Its square $g=\omega t^2$ certifies invertibility and gives the cofactor frame after normalization by $g^{-1}=\omega^2t^{46}$, although $C_{48}$ is not a $2$-group.

The $[[768,192,22]]$ entry with check weight $w=10$ has $\Delta_X=1+t^{15}+t^{18}+t^{39}+t^{42}$ and $\Delta_X^2=1$; its cofactor frame has width $29$.

Allowing one zero per quaternary check row gives two alternatives with check weight $w=8$: $[[768,192,13]]$ with $f_{\mathrm{rep}}=20$, and $[[768,192,12]]$ with $f_{\mathrm{rep}}=15$.
Their pivot determinants are $t^{16}+\omega^2(1+t^{24})$ and $\omega^2t^9+\omega t^{33}$, respectively, with squares $t^{32}$ and $t^{18}$; both normalizations preserve binary weight.
These two constructions and the example with $w=9$ have maximum column weight four.

\subsubsection{Check-weight variants at 1024 qubits}
\label{sec:weight-nine-quaternary}
\label{sec:weight-ten-quaternary}
A different period-$64$ choice gives $[[1024,256,18]]$ with $w=9$ and $f_{\mathrm{rep}}=26$.
Each coefficient row has six $1$'s, one $\omega$, and one $\omega^2$, so both trace components have check weight nine.
The first three augmentation columns have determinant $\omega$, ensuring full regularity.
The pivot minor and pairing are
\begin{align}
\Delta_9=\omega(t^2+t^{26}+t^{34}),\qquad \Delta_9^2=\omega^2t^{52}.
\end{align}
The raw cofactors have maximum binary weight $26$; normalization by $g^{-1}=\omega t^{12}$ mixes the trace components and gives a canonical frame of width $28$.
Index the four binary seed pairs $((x_1,z_1),\ldots,(x_4,z_4))$ in the order $(3,\omega)$, $(3,\omega^2)$, $(7,\omega)$, $(7,\omega^2)$ of free column and trace-basis element, and apply
\begin{align}
x_3\longleftarrow x_3+x_4,\qquad z_4\longleftarrow z_4+z_3.
\label{eq:weight9-frame-shear}
\end{align}
This paired basis change preserves every translated canonical pairing and gives a canonical frame of width $26$.

The $[[1024,256,24]]$ construction with $w=10$ has four $1$'s, two $\omega$'s, and two $\omega^2$'s in each coefficient row.
The minor
\begin{align}
\Delta_{10}=t^{31}+t^{60}+t^{63},\qquad \Delta_{10}^{2}=t^{56}
\end{align}
gives a full regular canonical frame after shifting the $Z$ seeds by $t^8$.
The normalized seeds have odd weights $25$ or $27$, so $f_{\mathrm{rep}}=27$.
Every $Z$ stabilizer has even weight, whereas the $Z$ seeds have odd weight.
The parity-functional argument in Appendix~\ref{sec:parity-functional} therefore gives
\begin{align}
25\le f\le f_{\rm reg}\le27.
\end{align}